\documentclass[12pt]{article}

\usepackage{mst-stylefile} 

\def\esssup{{\rm esssup}}
\def\eqref#1{(\ref{#1})}

 \usepackage{Bbb11a}
 \usepackage{natbib}
 \usepackage{times}
 \usepackage{amsfonts}

 \usepackage{graphicx}
 \usepackage{subfigure}

 \usepackage{color}

\usepackage{multirow}
 \usepackage{rotating}
\usepackage{graphics}

\include{epsf}

\renewcommand\cal{\mathcal }

\renewcommand{\Bbb}{\mathbb}

\newcommand{\bbZ}{{\Bbb Z}}
\newcommand{\bbR}{{\Bbb R}}

\newcommand{\bbN}{{\Bbb N}}

\newcommand\const{\mbox{\rm const}\,}

\newcommand\E{\Bbb E}
\renewcommand\P{\Bbb P}
\newcommand\noi{\noindent}

\newcommand{\what}{\widehat}

\def\D{{\cal D}([0,\infty),\bbR)}

\def\wtilde{\widetilde }

\def\D{D} 

\title{On the Estimation of the Heavy--Tail Exponent in Time Series using the Max--Spectrum\thanks{The authors were partially supported by
 NSF grant DMS-0806094.}}

\author{Stilian A.\ Stoev\thanks{ {\em Address for correspondence:} 
 Stilian A.\ Stoev, Department of Statistics, The University of Michigan, 439 West Hall, 
 1085 South University, Ann Arbor, MI 48109-1107, U.S.A. {\it E-mail:} sstoev@umich.edu.}{\, \,  and \ }George Michailidis}

\date{January 16, 2009 \\ {\it Department of Statistics, The University of Michigan, Ann Arbor, U.S.A.}}

\begin{document}

\maketitle

%
%
%
%

\begin{abstract}
This paper addresses the problem of estimating the tail index $\alpha$ of distributions 
with heavy, Pareto--type tails for dependent data, that is of interest in the areas of
finance, insurance, environmental monitoring and teletraffic analysis. A novel approach
based on the max self--similarity scaling behavior of block maxima is introduced. The method
exploits the increasing lack of dependence of maxima over large size blocks, which proves
useful for time series data. 

We establish the consistency and asymptotic normality of the proposed max--spectrum estimator 
for a large class of $m-$dependent time series, in the regime of intermediate block--maxima.
In the regime of large block--maxima, we demonstrate the distributional consistency of the estimator
for a broad range of time series models including linear processes. The max--spectrum estimator
is a robust and computationally efficient tool, which provides a novel time--scale perspective
to the estimation of the tail--exponents. Its performance is illustrated over synthetic and real data sets.


\end{abstract}

\section{Introduction}
 \label{s:intro}
 
The problem of estimating the exponent in heavy tailed data has a long history in statistics,
due to its practical importance and the technical challenges it poses. Heavy tailed distributions
are characterized by the slow, hyperbolic decay of their tail.  Formally, a real valued random 
variable $X$ 
with cumulative distribution function (c.d.f.) $F(x) = \P\{X\le x\},\ x\in\bbR$ is (right) heavy--tailed
with index $\alpha>0$, if 
\beq\label{e:F-tail}
  \P\{X>x\} = 1 - F(x) \sim L(x) x^{-\alpha},\ \mbox{ as } x\to\infty, 
\eeq
where $\sim$  means that the ratio of the left--hand side
to the right--hand side in \refeq{F-tail} tends to $1$, as $x\to\infty$.  Here
$L(\cdot)$ is a slowly varying function at infinity, i.e.\ $L(\lambda x)/L(x) \to 1,$ as $x\to\infty$,
for all $\lambda>0$.
For simplicity purposes, we suppose that  $X$ is almost surely positive i.e.\ $F(0)=0$, and we also focus on
the case when $L(\cdot)$ is asymptotically constant, namely
\beq\label{e:Lx}
 L(x) \sim \sigma_0^\alpha,\ \ \mbox{ as }x\to\infty,
\eeq
for some $\sigma_0>0$.  The case when
the slowly varying function $L(\cdot)$ is non--trivial is discussed in the Remarks after Theorem \ref{t:main}, below.

The {\it tail index} ({\it exponent}) $\alpha$ controls the rate of decay of the tail of $F$.
The presence of heavy tails in data was originally noted in the work of Zipf on word frequencies in 
languages (\cite{zipf:1932}), who also introduced a graphical device for their detection
(\cite{desousa:michailidis:2004}). Subsequently, \cite{mandelbrot:1960} noted their
presence in financial data. Since the early 1970s heavy tailed behavior has been noted
 in many other scientific fields, such as hydrology, insurance claims and social and biological networks
 (see, e.g.\ \cite{finkenstadt:rootzen:2004} and \cite{barabasi:2002}). In particular, the 
 emergence of the Internet and the World Wide
 Web gave a new impetus to the study of heavy tailed distributions, due to their omnipresence
 in Internet packet and flow data, the topological structure of the Web, the size of computer files, etc.
 (see e.g.\ \cite{adler:feldman:taqqu:1998},  \cite{resnick:1997}, \cite{faloutsos:faloutsos:faloutsos:1999},
 \cite{adamic:huberman:2000,adamic:huberman:2002}, \cite{park:willinger:2000}).
 In fact, heavy tailed behavior is a characteristic of highly optimized physical systems, as argued in
 \cite{carlson:doyle:1999}.
 
Heavy tails are also ubiquitous in stock market data. It is well--documented that the returns of many stocks measured
at high--frequency exhibit non--negligible extreme fluctuations, consistent with a non--Gaussian, heavy--tailed
model. The availability of high--frequency tic-by-tic data reveals further pronounced presence of heavy tails in
the transaction volumes. Figure \ref{fig:first-data} shows the volumes associated with all single transactions
of the Honeywell Inc.\ stocks recorded during January 4th, 2005 at the New York Stock Exchange (NYSE) and NASDAQ (see,
\cite{wharton.data.service}).  The transactions are ordered by their occurrence in time.  The presence of large
spikes indicates heavy tails, 
similar, for example, to the moving average with Pareto innovations shown in Figure \ref{fig:first-data-example} below.

Some important features of such data are: (i) their large size due to the fine time scale resolution (high--frequency)
at which they are collected (ii) their temporal structure that introduces dependence amongst observations, and 
(iii) their sequential nature, since observations are added to the data set over time. Traditional methods for
estimating the tail index are not well suited for addressing these issues, as discussed below.
\begin{figure}[t!]
\begin{center}
\includegraphics[width=4.5in]{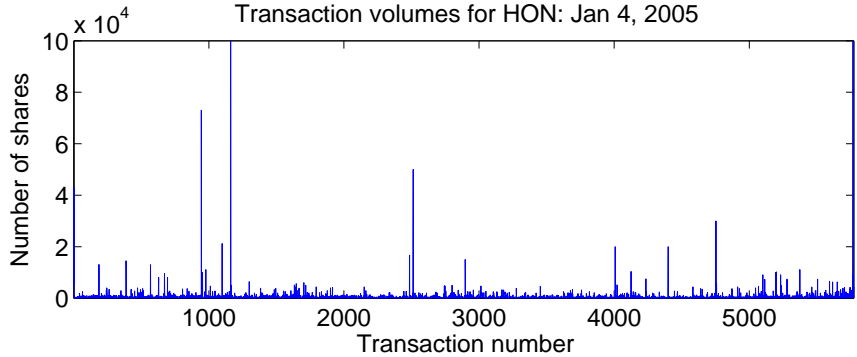}
{\caption{\label{fig:first-data} \small  
Transaction volumes for {\it Honeywell Inc.\ }(HON) from the NYSE and NASDAQ consolidated trades and quotes data base during January 4th, 2005.
The observations correspond to the volumes (in a number shares) per single transaction.  The transactions are listed in
the order of their occurrence in time.  The observed heavy--tailed behavior of traded volumes is ubiquitous across
different trading days and across the entire spectrum of relatively liquid stocks.}}
\end{center}
\end{figure}

The majority of the approaches proposed in the literature focuses on the scaling behavior 
of the largest order statistics $X_{(1)} \ge X_{(2)} \ge \cdots \ge X_ {(n)}$ obtained from
an in dependent and identically distributed (i.i.d.) sample $X(1),\ldots,X(n)$ from $F$; typical examples include
Hill's estimator \cite{hill:1975} and its numerous variations (\cite{kratz:resnick:1996},
\cite{resnick:starica:1997}), kernel based estimators (\cite{csorgo:deheuvels:mason:1985} and
\cite{feuerverger:hall:1999}). A review of these methods and their applications is given in
\cite{dehaan:drees:resnick:2000} and \cite{desousa:michailidis:2004}). 
 The most widely used in practice is the Hill estimator $\what \alpha_H(k)$ defined as:
\beq\label{e:alpha-k-Y}
 \what \alpha_H(k) := {\Big(} \frac{1}{k} \sum_{i=1}^k \ln X_{(i)} - \ln X_{(k+1)} {\Big)}^{-1},
\eeq 
with $k, 1\le k \le n-1$ being the number of included order statistics.  The parameter $k$ is
typically selected by examining the plot of the $\what \alpha_H(k)$'s versus $k$, known as the
{\it Hill plot}. In practice, one chooses a value of $k$ where the Hill plot exhibits a fairly constant
behavior (see e.g.\ \cite{dehaan:drees:resnick:2000}).
However, the use of order statistics requires
{\em sorting} the data that is computationally expensive (requires at 
least  $\mathcal{O}(n\log(n))$ steps) and destroys the time ordering of the data and
hence their temporal structure. Further, as can be seen from the brief
review above, most of the emphasis has been placed on point estimation of the tail index and
little on constructing confidence intervals. Exceptions can be found in the work 
of \cite{cheng:peng:2001} and \cite{lu:peng:2002} for the construction of confidence intervals and of 
\cite{resnick:starica:1995} on the estimation of $\alpha$ for dependent data.

\medskip
The purpose of this study is to introduce a method for estimating the tail index that overcomes the 
above listed shortcomings of other techniques. It is based on the asymptotic 
{\em max self--similarity} properties of heavy--tailed maxima. Specifically, 
the maximum values of data calculated over 
blocks of size $m$, scale at a rate of $m^{1/\alpha}$. Therefore, by examining a sequence
of growing, dyadic block sizes $m = 2^j,\ 1\le j \le \log_2 n,\ j\in\bbN$, and subsequently estimating the mean of
logarithms of block--maxima (log--block--maxima) one obtains an estimate of the tail index $\alpha$.
Notice that by using blocks of data, the temporal structure of the data is preserved. This procedure requires
$\mathcal{O}(n)$ operations, making it particularly useful for large data sets; further, the estimates for
$\alpha$ can be updated recursively as new data become available, by using only ${\cal O}(\log_2 n)$ memory
and without the knowledge of the entire data set, thus making the proposed estimator particularly suitable
for streaming data. Estimators based on max--self similarity for the tail index for i.i.d.\ data 
were introduced in \cite{stoev:michailidis:taqqu:2006}, where their consistency and 
asymptotic normality was established. In this paper, we extend them to dependent data, prove their
consistency, examine and illustrate their performance using synthetic and real data sets and discuss a
 number of implementation issues.

The remainder of the paper is structured as follows: in Section \ref{s:estimators_intro} 
the max--spectrum estimators are introduced.
Their consistency and asymptotic normality is established in Section \ref{s:an}, for $m-$dependent processes.
The distributional consistency of the estimators is established in Section \ref{s:d-cons} for a large 
class of time series models (including linear processes) under a mild asymptotic independence condition.
The construction of confidence intervals is further addressed in Section \ref{s:conf_int}.  
The important problem of automatic selection of parameters is addressed in Section \ref{s:cut-off}.
Applications to financial time series are discussed in Section \ref{s:data}, while most technical proofs are given in the Appendix.

\section{Max self--similarity and tail exponent estimators}
 \label{s:estimators_intro}

 Here we introduce the max self--similarity estimators for the tail exponent
 and demonstrate several of their characteristics. We start by reviewing the basic ideas for the case of
 independent and identically distributed  (i.i.d.) data. A detailed exposition is given in \cite{stoev:michailidis:taqqu:2006}.

Consider the sequence of block--maxima
$$
   X_m(k) := \max_{1\le i \le m}X(m(k-1)+i) = \bigvee_{i=1}^m X(m(k-1)+i),\ \ k=1,2,\ldots,~~m\in\bbN,
$$
where $X_m(k)$ denotes the largest observation in the $k-$th block. By \refeq{F-tail} \& 
\refeq{Lx} and the Fisher--Tippett--Gnedenko Theorem,
\beq\label{e:X_m}
   {\Big\{}\frac{1}{m^{1/\alpha}} X_m(k){\Big\}}_{k\in\bbN} \stackrel{d}{\longrightarrow}
   {\Big\{} Z(k) {\Big\}}_{k\in\bbN},\ \ \ \mbox{ as }m\to\infty,
\eeq
where $\stackrel{d}{\to}$ denotes convergence of the finite--dimensional distributions, with
the $Z(k)$'s being independent copies of an $\alpha-$Fr\'echet random variable.  
A random variable $Z$ is said to be $\alpha-$Fr\'echet, $\alpha>0$, with scale coefficient $\sigma>0$, if
\beq\label{e:Z}
 \P\{ Z \le x \} = \left\{\begin{array}{ll}
 \exp\{- \sigma^\alpha x^{-\alpha}\} &,\ x>0\\
 0 &,\ x \le 0
 \end{array} \right.
\eeq
The Fr\'echet variable $Z$ is said to be standard if $\sigma=1$.

Thus, for large $m$'s the block--maxima $X_m(k)$'s behave like a sequence of i.i.d.\ $\alpha-$Fr\'echet variables,
which suggests the following:

\begin{definition}\label{d:max-ss}
 A sequence of random variables $X = \{X(k)\}_{k\in\bbN}$ is said to be max self--similar with self--similarity 
 parameter $H>0$, if for any $m>0$,
\beq\label{e:d:max-ss}
 {\{}\bigvee_{i=1}^{m} X(m(k-1)+i) {\}}_{k\in\bbN} \stackrel {d}{=} {\{}m^H X(k){\}}_{k\in\bbN},
\eeq
 with $=^d$ denoting equality of the finite--dimensional distributions.
\end{definition}

Relationship  \refeq{d:max-ss} holds asymptotically for i.i.d. data and exactly for Fr\'echet distributed
data. Hence, any sequence of i.i.d.\  heavy--tailed variables can be regarded as {\it asymptotically max self--similar}
with self--similarity parameter $H=1/\alpha$.  This feature suggests that an estimator of $H$ and consequently $\alpha$ 
can be obtained by focusing on the scaling of the maximum values in blocks of growing size.  A similar
idea applied to block--wise sums was used in \cite{crovella:taqqu:1999} for estimating $\alpha$,
in the case $0<\alpha<2$.

\medskip
For  an i.i.d.\  sample $X(1),\ldots,X(n)$ from $F$, define
\beq\label{e:C-j}
  \D(j,k) := \max_{1\le i \le 2^j} X(2^j(k-1)+i) = \bigvee_{i=1}^{2^j} X(2^j(k-1)+i),\ \ k = 1,2,\ldots, n_j,
\eeq 
for all $j=1,2,\ldots,[\log_2 n],$  where $n_j := [n/2^j]$ and where $[x]$ denotes the largest integer not
greater than $x\in\bbR$. By analogy to the discrete wavelet transform, we refer to the parameter $j$ as the {\it scale} and
to $k$ as the {\it location} parameter.  We consider dyadic block--sizes because of their
algorithmic and computational advantages. Introduce the statistics
\beq\label{e:Yj}
  Y_j := \frac{1}{n_j} \sum_{k=1}^{n_j} \log_2 \D(j,k),\ \ j=1,2,\ldots,[\log_2 n].
\eeq
The Law of Large Numbers implies that for fixed $j$, as $n_j\to\infty$, the $Y_j$'s are consistent 
and unbiased estimators of $\E Y_j = \E \log_2 \D(j,1)$, if finite (see Corollary 3.1 in
\cite{stoev:michailidis:taqqu:2006}). On the other hand, the asymptotic max self--similarity
\refeq{X_m} of $X$ and \refeq{C-j} suggest that under additional tail regularity conditions (see e.g.\ Proposition \ref{p:moments} below):
\beq\label{e:Yj-sim}
  \E Y_j = \E \log_2 \D(j,1) \simeq j H + C \equiv j/\alpha + C, \ \ \ \mbox{ as }j\to\infty,
\eeq
where $C = C(\sigma_0,\alpha) = \E \log_2 \sigma_0 Z$, and where $\simeq$ means that the difference 
between the left-- and the right--hand side tends to zero, with $Z$ being an $\alpha-$Fr\'echet 
variable with unit scale coefficient. 

Then, a regression--based estimator of $H = 1/\alpha$ (and hence $\alpha$)
for a range of scales $1\le j_1\le j \le j_2 \le [\log_2 n]$ is given by:
\beq\label{e:H-hat}
 \what H(j_1,j_2) := \sum_{j=j_1}^{j_2} w_j Y_j,\ \ \mbox{ and }\ \ \what \alpha (j_1,j_2) := 1/\what H(j_1,j_2),
\eeq
where the weights $w_j$'s are chosen so that $\sum_{j=j_1}^{j_2} w_j = 0$ and $\sum_{j=j_1}^{j_2} j w_j = 1$.
The optimal weights $w_j$'s can be calculated through {\it generalized least squares} (GLS) regression using the
asymptotic covariance matrix of the $Y_j$'s. In practice, it is important to at least use {\it weighted least squares}
(WLS) regression to account for the difference in the variances of the $Y_j$'s (see, \cite{stoev:michailidis:taqqu:2006}).
\begin{figure}[h!]
\begin{center}
\includegraphics[width=4.5in,height=3in]{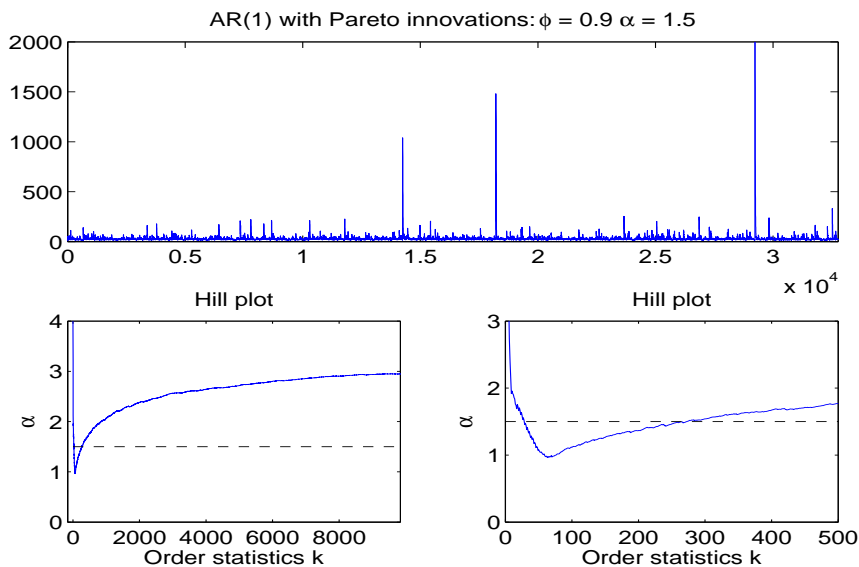}
{\caption{\label{fig:first-data-example} \small
{\it Top panel}: auto--regressive time series of order 1 with Pareto innovations of tail exponent $\alpha=1.5$.
{\it Bottom left and right panels}: the Hill plot for this data set and its zoomed--in version, respectively.
 The dashed horizontal line indicates the value of $\alpha = 1.5$.
}}
\end{center}
\end{figure}

\begin{figure}[h!]
\begin{center}
\includegraphics[height=2.5in,width=2.5in]{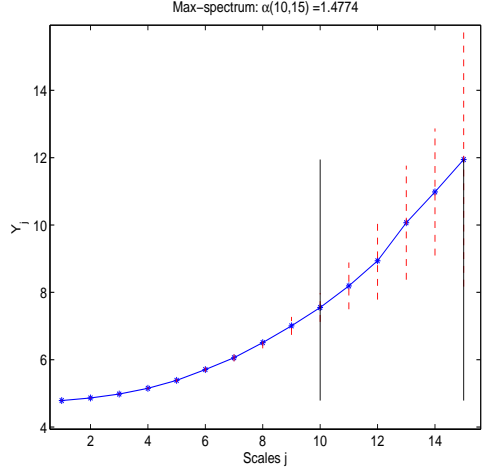}
{\caption{\label{fig:first-data-example-mspec} \small
The max--spectrum plot for the data set in Figure \ref{fig:first-data-example}.  The max self--similarity
estimator of the tail exponent, obtained from the range of scales $(j_1,j_2) = (10,15)$, is $\what \alpha(10,15) = 1.4774$.
}}
\end{center}
\end{figure}

\medskip
{\it We propose to use the estimator defined in \refeq{H-hat} for dependent time series data.}
We first illustrate its usage through a simulated data example. A data set of size $n = 2^{15} = 32,768$ was generated
from an auto--regressive time series of order one with Pareto innovations. Specifically,
$$
  X(k) = \phi X(k-1) + Z(k) = \sum_{i=0}^\infty \phi^i Z(k-i),\ \ k=1,\ldots, n,
$$
where $\phi = 0.9$ and $\P\{Z(k)>x\} = x^{-\alpha},\ x>1$, with $\alpha = 1.5$.  The data together with
its Hill plot are shown in Figure \ref{fig:first-data-example}.  Notice that even though the Hill estimator
work best for Pareto data, the dependence structure in the model leads to a Hill plot, which is substantially different 
from that for independent Pareto data (see the bottom left panel).  The zoomed--in version of the Hill plot (bottom right 
panel) however indicates that the tail exponent should be in the range between $1$ and $2$.  The choices of $k$ in 
the range between $200$ and $400$ do in fact lead to estimates around $1.5$.  This range however is hard to guess 
if one did not know the true value of $\alpha = 1.5$. 
\cite{resnick:starica:1997} have shown that the Hill estimator is consistent for such 
dependent data sets. Nevertheless, as this example indicates, the Hill plot can be difficult to
assess in practice.  

In Figure \ref{fig:first-data-example-mspec}, the {\it max--spectrum} plot is shown; i.e.
the plot of the statistics $Y_j$ versus the available dyadic scales $j,~1\leq j\leq [\log_2 n] (=15)$. 
The estimated tail exponent over the range of scales $(10,15)$ is
$1.4774$, which is very close to the nominal value of $\alpha= 1.5$.  Moreover, the max--spectrum
is easy to assess and interpret.  One sees a ``knee'' in the plot near scale $j=10$, where the 
max--spectrum curves upwards and thus it is natural to choose the range of scales $(10,15)$ to 
estimate $\alpha$.  The choice of the scales $(j_1,j_2)$ can be also automated, as briefly discussed in 
Section \ref{s:cut-off} below.

\medskip
\noi{\bf Remark:} {\it (on the algorithmic implementation)}
 The max--spectrum $Y_j,\ j=1,\ldots,[\log_2 n]$ of a data set $X_1,\ldots, X_n$ can be computed efficiently 
in ${\cal O}(n)$ steps, without sorting the data. Indeed, this is evident from the recursive
construction of block maxima, since 
$$
  D(j+1,k) = \max \{ D(j,2k-1), D(j, 2k)\},\ \ k=1,\ldots,[n/2^{j+1}],\ \ 1\le j <[ \log_2 n].
$$
Moreover, this property can be further used to obtain a {\it sequential algorithm} for the computation of the
$Y_j$'s.  Indeed, keep in addition to the $Y_j$'s, the last block--maximum $D_j:= D(j,n_j),\ n_j = [n/2^j]$ per scale $j$,
and also the extra variables  $R_j = \max_{2^j n_j < i \le n} X(i)$, which represent the  maxima of
the  'left--over' $X_i$'s over the range  $ 2^j [n/2^j] <i \le n$.  Now, if a new observation $X_{n+1}$
is recorded, one can easily update the $R_j$'s and the $Y_j$'s, with the help of the $R_j$'s, and 
the $D_{j'}:= D(j',n_{j'})$'s, for $1\le j'< j$.  Thus, one recovers the $(Y_j, D_j, R_j)_{j}-$representation of the data
$X_1,\ldots, X_{n+1}$.  Since only $\log_2 n$ scales are available, we perform ${\cal O}(\log_2 n)$ operations per 
update and use ${\cal O}(\log_2 n)$ memory to store the max--spectrum and the auxiliary data.

This sequential implementation of the max--spectrum is of critical importance in the context of 
data streams in modern data bases or Internet traffic applications.  In such settings, large volumes
of data are observed in short amounts of time; they cannot be stored and/or sorted efficiently while at
the same time rapid 'queries' need to be answered about various statistics of the data.  The proposed
max--spectrum estimator provides a unique tool for the estimation of the tail--exponent of such data.
Notice that the other available techniques require sorting the data which is impossible without having to store
the entire data set. A sequential implementation of the Hill estimator for example would require ${\cal O}(n)$ memory,
which is prohibitive in many applications.

\section{Asymptotic properties}

\subsection{Asymptotic Normality  (in the intermediate scales regime) }
\label{s:an}

 The estimators $\what H $ and $\what \alpha = 1/\what H$ in \refeq{H-hat} utilize the scaling
 properties of the max--spectrum statistics $Y_j$ in \refeq{Yj}.
 The discussion in Section \ref{s:estimators_intro} suggests that the max self--similarity estimators 
 in \refeq{H-hat} will be consistent as both the scale $j$ and $n_j$ tend to infinity.  The consistency and
 asymptotic normality of these estimators was established in \cite{stoev:michailidis:taqqu:2006} for i.i.d.\ data.
 This was accomplished by assessing the rate of convergence of moment type functionals of block--maxima, such as 
 $\E \log_2 D(j,1)$, under mild conditions on the rate of the tail decay in \refeq{F-tail}.
 Here, we focus on the case of dependent data and establish the asymptotic normality of the proposed
 max self--similarity estimators under analogous conditions on the rate.

\medskip
 Consider a strictly stationary process (time series)  $X= \{X(k)\}_{k\in\bbZ}$ with heavy--tailed
 marginal c.d.f.\ $F$ as in \refeq{F-tail} \& \refeq{Lx}. Further, assume that the $X(i)$'s are positive, almost surely,
 that is, $F(0)=0$.  In many contexts, the block--maxima of $X$ scale at a rate $m^{1/\alpha}$ as the
 block size $m$ grows even under the presence of strong dependence.  This is so, for example, when the
 time series $X$ has a {\it positive extremal index} (see, p.\ 53 in \cite{leadbetter:lindgren:rootzen:1983}).
 The following conditions make this more precise by quantifying further the rate of convergence.

 Let $M_n := \max_{1\le k\le n} X(k)$ and let
$$
 F_n(x)\equiv \P \{  M_n/n^{1/\alpha} \le x\} =: \exp\{ -c(n,x) x^{-\alpha}\},\ \ \ x >0,\ \ n\in\bbN. 
$$ 
One can see that 
$M_n/n^{1/\alpha}\stackrel{d}{\longrightarrow} Z, \ n\to\infty$ if and only if  
$c(n,x) \to c_X \equiv \const,\ n\to\infty,$ for all $x>0$, where $Z$ is an $\alpha-$Fr\'echet variable with
scale $\sigma = c_X^{1/\alpha}$.
The following conditions will help us quantify the rate of the last convergence and also obtain
rates of convergence for moment functionals of block--maxima in Proposition \ref{p:log-moments} below.

\medskip
\begin{condition}\label{cond:C1} There exists $\beta>0$ and $R \in \bbR$, such that
\begin{equation} \label{e:C1}
|c(n,x) - c_X| \le c_1(x) n^{-\beta},\ \ \mbox{ for all } x >0,\ \ \ \ \mbox{ and }\ \ \ c_1(x) = {\cal O}(x^{-R}),\ x\downarrow 0,
\end{equation}
for some $c_X>0$.
\end{condition}

\medskip
\begin{condition}\label{cond:C2} For all $x>0$, we have
\begin{equation} \label{e:C2}
 c(n,x) \ge c_2 \min\{1,x^\gamma\},\ \ \mbox{ for some }\gamma\in (0,\alpha),
\end{equation}
for all sufficiently large $n \in \mathbb{N}$, where $c_2>0$ does not depend on $n$.
\end{condition}

\medskip
\noi{\bf Remarks}
\begin{enumerate}
 \item [1.] Conditions \ref{cond:C1} and \ref{cond:C2} are not very restrictive.  They can be
 shown to hold, for example, for a large class of 
moving maxima processes (see, Proposition \ref{p:m_max}).

 \item [2.] Condition \ref{cond:C1} implies in particular that $M_n/n^{1/\alpha} \stackrel{d}{\to} c_X^{1/\alpha} Z,$ as $n\to\infty$,
 for a standard $\alpha-$Fr\'echet variable $Z$.  In view of \refeq{F-tail}, we also have that $M_{n}^*/n^{1/\alpha}  \stackrel{d}{\to} 
 \sigma_0 Z$, as $n\to\infty$, with $M_n^* := \max_{1\le k\le n} X(k)^*$, where the $X(k)^*$'s are i.i.d.\ random variables with c.d.f.\ $F$.  This implies
 that the {\it extremal index} $\theta$ of the time series $X$ is: $\theta:= c_X/\sigma_0^\alpha$ (see, e.g.\ 
 p.\ 53 in \cite{leadbetter:lindgren:rootzen:1983}).
\end{enumerate}

\noi Conditions \ref{cond:C1} \& \ref{cond:C2}, yield the following
important result on the rate of convergence of log--block maxima, similar to Corollary 3.1 in
\cite{stoev:michailidis:taqqu:2006}.

\begin{proposition} \label{p:log-moments} Let $X = \{X(k)\}_{k \in \mathbb{Z}}$ be a strictly stationary time
 series which satisfies Conditions \ref{cond:C1} \& \ref{cond:C2}. Suppose that $\int_1^\infty c_1(x) x^{-\alpha-1 + \delta} dx <\infty$,
 for some $\delta>0$.
 Then, with $M_n := \max_{1\le k\le n} X_k$, we have $\E |\ln(M_n)|^p <\infty$,  for all $p>0$ and
 all sufficiently large $n\in\bbN$. Moreover, for any $p>0$ and $k\in\bbN$, we have:
$$
 {\Big|} \E |\ln(M_n/n^{1/\alpha})|^p - \E |\ln(Z)|^p {\Big|} = {\cal O}(n^{-\beta}),\ \ \mbox{ and }\ \
 {\Big|} \E (\ln(M_n/n^{1/\alpha}))^k - \E (\ln(Z))^k {\Big|} = {\cal O}(n^{-\beta}),
$$
as $n\to\infty$, where $Z$ is an $\alpha-$Fr\'echet random variable with scale coefficient $\sigma = c_X^{1/\alpha}$.
\end{proposition}

\noi The proof is given in Section \ref{s:proofs}. Proposition \ref{p:log-moments} readily implies:
\begin{equation}\label{e:E-Yj-rate}
 \E (Y_j - j/\alpha ) \equiv \E \log_2( D(j,k)/2^{j/\alpha} )    = \E \log_2( c_X^{1/\alpha} Z_1 ) + {\cal O}(1/2^{j\beta}),
\end{equation}
as $j\to\infty$, where $Z_1$ is a standard $\alpha-$Fr\'echet variable.
This result yields an asymptotic bound on the bias of the estimators $\what H(j_1,j_2)$ in \refeq{H-hat} above.
 
Proposition \ref{p:log-moments} can be further used to establish the asymptotic normality of $\what \alpha(j_1,j_2)
= 1/\what H(j_1,j_2)$ in \refeq{H-hat}. To do so, we focus on a range of scales
$(j_1,j_2)$ which grows with the sample size.  Namely, we fix $\ell\in \bbN,\ \ell \ge 2$, let
$j_1:= 1+ j(n)$ \& $j_2:= \ell + j(n)$, and as in \refeq{H-hat} define:
$$
 \what H_n := \sum_{i=1}^{\ell} w_i Y_{i+j(n)}\ \ \ \mbox{ and } \ \ \ \what \alpha_n := 1/\what H_n,
$$
where $\sum_{i=0}^{\ell-1} i^{\kappa} w_i  = \kappa,\ \kappa = 0,1$.
The next theorem is the main result of this section.  It establishes the asymptotic normality of
the estimator $\what \alpha_n$, as $j(n)$ and $n$ tend to infinity.

\begin{theorem}\label{t:main} Let $X_1,\ldots,X_n$ be a sample from an $m-$dependent heavy--tailed process 
$X = \{X(k)\}_{k\in\bbZ}$. Suppose that \refeq{C1} and \refeq{C2} hold and let $j=j(n) \in\bbN$ be such that
\beq\label{e:j(n)}
  2^{j(n)}/n + n/2^{j(n)(1+2\min\{1,\beta\})}  \longrightarrow 0,\ \ \ \mbox{ as }n\to\infty.
\eeq
Then, as $n\to\infty$,
\beq\label{e:t:main}
 \sqrt{n_j} (\what \alpha_n - \alpha) \stackrel{d}{\longrightarrow} {\cal N}(0,\alpha^2 c_w),\ \  \ \mbox{ with }
  c_w = \sum_{i',i''=1}^{\ell}w_{i'} w_{i''} \Sigma_1(i',i''),
\eeq
where $n_j= [n/2^j]$. Here $\Sigma_1(i',i'') = 2^{\min\{i',i''\}} {\rm Cov}(\log_2 Z_1, \log_2(Z_1 \vee (2^{|i'-i''|}-1)Z_2)),$
for $i',i''=1,\ldots,\ell$, where $Z_1$ and $Z_2$ are independent standard $1-$Fr\'echet random variables.
\end{theorem}
\begin{proof} By the 'Delta--method' (see e.g.\ Theorem 3.1 in \cite{vaart:1998}), it suffices to show that
 \beq\label{e:main-1}
 \sqrt{n_j} (\what H_n - H) \stackrel{d}{\longrightarrow} {\cal N}(0, H^2 c_w),\ \ \ \mbox{ as }n\to\infty.
 \eeq
Indeed, since $\what\alpha_n = f(\what H_n),$ with $f(x) = 1/x$, we have
$\what \alpha_n - \alpha = - H^{-2} (\what H_n - H) + o_p(\what H_n - H),$ as $n\to\infty$.

Let now 
$$
 \wtilde D(j',k) := \bigvee_{i=1}^{2^{j'} - m} X(2^{j'}(k-1) + i) \ \ \ \mbox{ and }\ \ \ \wtilde Y_{j'}:= \frac{1}{n_{j'}}\sum_{k=1}^{n_{j'}}
 \log_2 \wtilde D(j',k),  
$$
for all $j_1\le j' \le j_2$, and $k= 1,\ldots, n_{j'},$ where $n_{j'}:= [n/2^{j'}]$.  Observe that since the time series
$X$ is $m-$dependent, the $\wtilde D(j',k)$'s are now independent in $k$.  Hence, in view of Conditions \ref{cond:C1} \& \ref{cond:C2}
and Proposition \ref{p:log-moments}, the results of Theorem 4.1 in \cite{stoev:michailidis:taqqu:2006} readily apply to
the max--spectrum $\wtilde Y_{j'},\ j'=j_1,\ldots, j_2$, which is based on the independent $\wtilde D(j',k)$'s.
Therefore, by setting
$\wtilde H_n := \sum_{i=1}^{\ell} w_i \wtilde Y_{i+j(n)}$, we obtain:
\beq\label{e:main-2}
 \sqrt{n_j} (\wtilde H_n - H) \stackrel{d}{\longrightarrow} {\cal N}(0, H^2 c_w),\ \ \ \mbox{ as }n\to\infty.
\eeq
In view of \refeq{main-2}, to establish \refeq{main-1}, it is enough to show that $\what H_n - \wtilde H_n = o_p(1/\sqrt{n_j}),\ n\to\infty$,
or that, for example, 
\beq\label{e:main-3}
 \E (\what H_n -\wtilde H_n)^2 = {\rm Var}(\what H_n - \wtilde H_n) + (\E \what H_n - \E \wtilde H_n)^2 = o(1/n_j),\ \ \ 
 \mbox{ as }n\to\infty.
\eeq
{\it Consider first the term ${\rm Var}(\what H_n - \wtilde H_n)$.} Since 
$
 \what H_n - \wtilde H_n = \sum_{i=1}^{\ell} w_i (Y_{i+j(n)} - \wtilde Y_{i+j(n)}),
$
we have
$$
 {\rm Var}(\what  H_n - \wtilde H_n) \le C \sum_{i=1}^{\ell} {\rm Var}(Y_{i+j(n)} - \wtilde Y_{i+j(n)}) \le \frac{C'}{n_j} \sum_{i=1}^{\ell}
 {\rm Var} (\log_2 D(i+j(n),1) - \log_2 \wtilde D(i+j(n),1)),
$$
for some constants $C$ and $C'$, where the last inequality follows from Lemma \ref{l:m-dep-1}.
Now, Lemmas \ref{lemma:convergence_of_ratio_of_BM_to_degenrate_RV} and \ref{lemma:uniform_integrable_of_logratios} imply 
that ${\rm Var}(\log_2 D(i+j(n),1) - \log_2 \wtilde D(i+j(n),1))\to 0,\ n\to\infty$, and hence
${\rm  Var}(\what  H_n - \wtilde H_n) = o(1/n_j)$, as $n\to\infty$.

{\it Now, focus on the term $(\E \what H_n - \E \wtilde H_n)^2$ in \refeq{main-3}.} For some constant $C_\ell>0$, we have
\beq\label{e:main-4}
(\E \what H_n - \E \wtilde H_n)^2 \le C_\ell \sum_{i=1}^{\ell} (\E Y_{i+j(n)} - \E \wtilde Y_{i+j(n)} )^2 = C_\ell
 \sum_{i=1}^{\ell} {\Big(}\E \log_2 \frac{D(i+j(n),1)}{ \wtilde D(i+j(n),1)} {\Big )}^2,
\eeq
where we used the inequality $(\sum_{i=1}^\ell x_i)^2 \le \ell \sum_{i=1}^\ell x_i^2$ and the stationarity (in $k$) of the $D(i+j(n),k)$'s and
$\wtilde D(i+j(n),k)$'s.  We further have that
\beq\label{e:main-5}
 \E \log_2 \frac{D(i+j(n),1)}{\wtilde D(i+j(n),1)} = \E \log_2{\Big(}\frac{ D(i+j(n),1)}{ 2^{j(n)/\alpha}}{\Big)} -
 \E \log_2 {\Big(} \frac{\wtilde D(i+j(n),1)}{(2^{j(n)} -m)^{1/\alpha}} {\Big)} 
  - \frac{1}{\alpha} \log_2(1 - \frac{m}{2^{j(n)}}).
\eeq
Relation \refeq{E-Yj-rate} implies that the last two expectation are both equal to $\E \log_2(Z) + {\cal O}(1/2^{j(n)\beta}),$ as
$n\to\infty$, where $Z$ is an $\alpha-$Fr\'echet variable with scale $c_X^{1/\alpha}$.   Therefore, from \refeq{main-4} and \refeq{main-5},
we obtain
$$
(\E \what H_n - \E \wtilde H_n)^2 = {\cal O}(1/2^{2j(n)\beta }) + {\cal O}(1/2^{2j(n)})
 = {\cal O}( 1/2^{2j(n) \min\{1,\beta\}}),\ \  \ \mbox{ as }n\to\infty,
$$
where in the last relation we used that $\log_2(1-x) = {\cal O}(x),\ x\to 0$.

By combining the above derived bounds on the terms on the right--hand side of \refeq{main-3}, we 
get
$$
\what H_n - \wtilde H_n = o_P(1/\sqrt{n_j}) + o_p( 1/2^{j(n) \min\{1,\beta\}}) = o_p(1/\sqrt{n_j}),\ \ \ \mbox{ as }n\to\infty,
$$
where the last equality follows from \refeq{j(n)} since $n/2^{j(n)(1+2\min\{1,\beta\})} = n_j/2^{2j(n)\min\{1,\beta\}} \to 0,$
as $n\to\infty$. This implies \refeq{main-1} and completes the proof of the theorem. 
$\Box$
\end{proof}

\medskip

We conclude this section with several important remarks on the scope of validity of the asymptotic results in Theorem \ref{t:main}.

\medskip
\noi{\bf Remarks}
\begin{enumerate}
\item
 {\it (On the role of $\beta$)}  The parameter $\beta>0$ in Condition \ref{cond:C1} controls the rate of the convergence in distribution
of $M_n/n^{1/\alpha}$ to the $\alpha-$Fr\'echet limit law.  The larger the value of $\beta$, the faster the convergence in \refeq{C1}, and in view of \refeq{j(n)},
the wider the range of scales $j(n)$'s in Theorem \ref{t:main} that lead to asymptotically normal $\what \alpha_n$'s.  In particular, the larger the $\beta$, the
faster the convergence of the $\what \alpha_n$'s can be made, since one could choose relatively small $j(n)$'s.  

On the other hand, when the rate of convergence of the law of $M_n/n^{1/\alpha}$ to its limit is relatively slow, then the values of $\beta>0$ can be close to zero.
This can lead to arbitrarily slow rates of the convergence of $\what \alpha_n$ since one may have to choose relatively large scales $j(n)$'s to compensate for the
rate of the bias in the max--spectrum on smaller scales.

 \item {\it (On the connection with Hill estimators)}
 As argued in \cite{stoev:michailidis:taqqu:2006}, for the case of independent data, Condition \ref{cond:C1} corresponds precisely
 to the second--order condition used in \cite{hall:1982}, where the asymptotic normality of the Hill estimator was established.  The rates of 
 convergence in \refeq{t:main} above are, in the case of independent data, in close correspondence with the rates for the Hill estimator, obtained in
 \cite{hall:1982}. 

\item {\it (On data with regularly varying tails)} Consider the case when the $X(k)$'s satisfy \refeq{F-tail} where now the slowly varying function 
$L(\cdot)$ is {\it non--trivial}.  Then, the max--spectrum based estimators of $\alpha$ will continue to work.  Indeed,
for the case of i.i.d.\ data, we have that
$$
 \frac{M_n}{a_n} \stackrel{d}{\longrightarrow} Z,\ \ \mbox{ as }n\to\infty,
$$
where $Z$ is a standard $\alpha-$Fr\'echet variable, and where $a_n = n^{1/\alpha}\ell(n)$ is such that 
$$
 na_n^{-\alpha} L(a_n) \equiv \frac{1}{\ell(n)} L(n^{1/\alpha}\ell(n)) \longrightarrow 1,\ \ \mbox{ as }n\to\infty.
$$
Here $\ell(\cdot)$ is another slowly varying function related to $L$ (see e.g.\ Proposition 1.11 in \cite{resnick:1987}).

If one replaces $n^{1/\alpha}$ by $a_n\equiv n^{1/\alpha}\ell(n)$ and $c_X$ by $1$ in Conditions \ref{cond:C1} and \ref{cond:C2} then
Propositions \ref{p:log-moments} and \ref{p:moments} will continue to hold with $M_n/n^{1/\alpha}$ replaced by $M_n/a_n$. 
The proofs are essentially the the same.  In the case of independent data, one has that
\beq\label{e:Var}
 {\rm Var}(Y_j) = \frac{1}{n_j} {\rm Var}{\Big(} \log_2 D(j,1)/(2^{j/\alpha} \ell(2^{j})^{1/\alpha}) {\Big)} = \frac{1}{n_j} {\Big(} {\rm Var}(\log_2 Z) + o(1) {\Big)},
\eeq
where the remainder term $o(1)$ vanishes, as $j\to\infty$, because of the analog of Relation \refeq{E-Yj-rate}.  Also, by the counterpart of 
\refeq{E-Yj-rate}, one obtains:
\beq\label{e:mean}
 \E {\Big (} Y_j - j/\alpha - \log_2(\ell(2^{j}))/\alpha{\Big)}  = \E \log_2(Z) + {\cal O}(1/2^{j\beta}),\ \ \mbox{ as }j\to\infty.
\eeq
Consider now a fixed $m\in\bbN,\ m\ge 2$ and let
$$
 \what H_n = \sum_{i=1}^m w_i Y_{i+j(n)}.
$$
Relation \refeq{Var} implies that ${\rm Var}(\what H_n)\to 0$, as $n$ and $j(n)$ tend to infinity.  On the other hand, Relation \refeq{mean} shows that
\begin{eqnarray}\label{e:ell}
 \E \what H_n &=& \frac{1}{\alpha}\sum_{i=1}^m i w_i + {\cal O}(1/2^{j\beta}) + 
(j(n) + \E \log_2(Z))\sum_{i=1}^m w_i + \frac{1}{\alpha} \sum_{i=1}^m \log_2(\ell(2^{i+j}))\nonumber\\
              &=& \frac{1}{\alpha} + {\cal O}(1/2^{j\beta}) + \frac{1}{\alpha} \sum_{i=1}^m \log_2 {\Big(} \ell(2^{i}\cdot 2^{j(n)})/ \ell(2^{j(n)}) {\Big)},
\end{eqnarray} 
where in the last two relations we used the facts that $\sum_{i=1}^m i w_i =1$ and  $\sum_{i=1}^m w_i = 0$.

Now, the fact that $\ell(\cdot)$ is a slowly varying function, implies that $\ell(2^{i} \cdot 2^{j(n)})/\ell(2^{j(n)}) \to 1$, as $j(n)\to\infty$.
This shows that the right--hand side of \refeq{ell} converges to $H\equiv 1/\alpha$, as $j(n)\to \infty$ and hence the estimator $\what H_n$ is consistent,
as $n\to\infty$ and as $j(n)\to\infty$.  Note that the rate of the bias $(\E \what H_n - H)$ depends not only on the term ${\cal O}(1/2^{j\beta})$ but also 
on the rate of the convergence
 $$
 \ell (\lambda j)/\ell(j) \longrightarrow 1,\ \ \mbox{ as }j\to\infty.
 $$ 
This last rate depends on the structure of the slowly varying function $\ell(\cdot)$ and it may be possible to control in terms of the
Karamata's integral representation 
$$
 \ell(x) = c_\ell(x) \exp{\Big\{} -\int_{x_0}^x \epsilon(u)/u du{\Big\}},
$$
at the expense however, of two additional parameters controlling the rates of $c_\ell(\cdot)$ and $\epsilon(\cdot)$.

This argument shows the consistency of the max--spectrum based estimator of $\alpha$ for i.i.d.\ $X(k)$'s with regularly varying tails. 
In principle, one can establish asymptotic normality of these estimators along similar line,  but this would involve 
technically complicated assumptions on the slowly varying functions considered.  
Further, as in Theorem \ref{t:main} one can establish asymptotic normality results for the $\what \alpha_n$'s
for $m-$dependent data.  We chose not to pursue the general case of regularly varying tails here since the technical details may obscure the idea behind the estimator.  These important theoretical results will be pursued in subsequent work on the subject.

\end{enumerate}

\subsection{Distributional consistency (in the large scales regime) }
 \label{s:d-cons}

 In Theorem \ref{t:main}, we consider an asymptotic regime where the number of block--maxima $n_j$ on 
 the scale $j=j(n)$ grows, as $n\to\infty$.  This is essential for the consistency of the estimators $\what \alpha_n$.
 In practice, however, the situation where we have a fixed number of block--maxima per scale is 
 also of interest.  Namely, for a sample $X(1),\ldots,X(n)$ and a {\it fixed} number of block--maxima $r$, we
 let $j(n):= [\log_2(n/r)]$ and consider the estimator 
 \beq\label{e:alpha-n-ell}
  \what\alpha_n := \sum_{i=1}^{\ell} w_i Y_{i+j(n)},\ \ \ \mbox{ where }\ \ell = [\log_2 r].
 \eeq
 This estimator corresponds to taking the largest $\ell$ scales in the max--spectrum, where $\ell$ is {\it fixed}.
 One cannot expect the estimators $\what \alpha_n$ to be consistent (even for independent data) since they involve
 averages over a fixed number of block--maxima statistics. Nevertheless, the asymptotic distribution of $\hat\alpha_n$
 is of interest.

 The next result establishes the 'distributional consistency' of the estimators $\what\alpha_n$ in the aforementioned 
 regime.  We do so under the condition that the block--maxima in \refeq{X_m} are asymptotically independent.
 This condition is in fact quite mild, as shown in Lemmas \ref{l:lin-proc} and \ref{l:ai} below.

 \begin{theorem}\label{t:alpha_n-to-alpha_F}  Suppose that \refeq{X_m} holds where the $Z(k)$'s are i.i.d.\  $\alpha-$Fr\'echet.
  Then,
 \beq\label{e:p:alpha_n-to-alpha_F}
 \what \alpha_n \stackrel{d}{\longrightarrow} \what \alpha_Z,\ \ \ \mbox{ as }n\to\infty,
 \eeq
 where $\hat \alpha_F:= 1/ (\sum_{i=1}^{\ell} w_i Y_{i}^{Z})$, and where $\{Y_i^Z\}_{i=1}^{\ell}$ is 
 the max--spectrum of a sequence of i.i.d.\  $\alpha-$Fr\'echet variables $Z(1),\ldots,Z(r)$.
 \end{theorem}
 \begin{proof}  The result readily follows from the continuous mapping theorem.  Indeed, by \refeq{X_m}, 
 and in view of \refeq{C-j}, we have 
 $
 \{ D(j(n),k)/2^{j(n)/\alpha},\ k=1,\ldots,r\} \stackrel{d}{\longrightarrow} \{ Z(k),\ k=1,\ldots,r\},\ \ \mbox{ as }n\to\infty,
 $
 and hence
 $$
 \{ \log_2 D(j(n),k) - j(n)/\alpha,\ k=1,\ldots,r\} \stackrel{d}{\longrightarrow} \{ \log_2 Z(k),\ k=1,\ldots,r\},\ \ \mbox{ as }n\to\infty.
 $$
 Due to the dyadic structure of the block--maxima $D(j(n),k)$'s, one can recover $D(i+j(n),k)$'s, for $i=1,\ldots,\ell$ and
 $k=1,\ldots, [r/2^{i}]$ from the block--maxima $D(j(n),k)$, $k=1,\ldots,r$ through a {\it continuous} combination of maxima operations.
 Thus, by applying the continuous mapping theorem again, we obtain
 $$
  \{ Y_{i+j(n)} - j(n)/\alpha \}_{i=1}^\ell \stackrel{d}{\longrightarrow}  \{ Y_{i}^Z \}_{i=1}^\ell,\ \ \mbox{ as }n\to\infty,
 $$
 which yields the convergence \refeq{p:alpha_n-to-alpha_F} since $\sum_{i=1}^\ell w_i j(n)/\alpha = 0$.
$\Box$
\end{proof}

 \medskip
 Condition \refeq{X_m} appears stringent, but contrary to intuition, it holds in most practical situations.
 We were unable to find an example of ergodic heavy--tailed time series $X$ (of positive extremal index) with 
 asymptotically {\it dependent} block--maxima.  We next show that \refeq{X_m} holds for 
 the large class of linear processes. 

 Let $\xi_k,\ k\in\bbZ$ be i.i.d.\   heavy--tailed innovations, such that 
 $\P\{ |\xi_k| > x\} \sim \sigma_\xi^\alpha x^{-\alpha},\ x\to\infty$, where 
 $\P\{ \xi_k > x\}/ \P\{ |\xi_k| > x\} \to p,\ x\to\infty$ for some $p\in [0,1]$.  
 Consider the linear process
 \beq\label{e:lin-proc}
  X(k) := \sum_{i=-\infty}^\infty c_i \xi_{k-i},\ \ k\in \bbZ.
 \eeq
\cite{mikosch:samorodnitsky:2000} provide a recent and comprehensive treatment
of the linear processes as in \eqref{e:lin-proc} (see also \cite{davis:resnick:1985}).
More precisely, by Lemma A.3 in \cite{mikosch:samorodnitsky:2000}, the following 
conditions on the $c_i$'s guarantee the almost sure convergence of the series in \eqref{e:lin-proc}. 
\beq\label{e:lin-proc-cond}
 \sum_{i=-\infty}^\infty c_i^2 < \infty \mbox{ (if $\alpha>2$) }\ \ \ \mbox{ and } \ \ \ \sum_{i=-\infty}^\infty
  |c_i|^{\alpha -\epsilon } <\infty,\ \ \mbox{ (if $\alpha\le 2$), for some }\epsilon>0.
\eeq
These conditions are necessary for $\alpha>0$, and nearly optimal for $\alpha\le 2$ (see Lemma A.3 in 
\cite{mikosch:samorodnitsky:2000}).  By Lemma A.3 in the last reference, we also have that the
tails of the $X(k)$'s are regularly varying with exponent $\alpha$ (see Relation (A.2) therein).

The following result shows that \refeq{X_m} holds for the linear process $X = \{X(k)\}_{k\in\bbZ}$ under the conditions
\eqref{e:lin-proc-cond}. The proof follows by a simple combination of arguments in \cite{davis:resnick:1985} and 
\cite{mikosch:samorodnitsky:2000} and it is given in the Appendix, for completeness. 

 \begin{lemma}\label{l:lin-proc} Let $c_+:= \max_{i\ge 0} c_i$ and $c_-:= \max_{i\ge 0} (-c_i)$. 
 Suppose that either $p c_+ >0$ or $(1-p)c_- >0$. Then, the linear process in \refeq{lin-proc} satisfies 
 \refeq{X_m} where the $Z(k)$'s are i.i.d.\  $\alpha-$Fr\'echet with scale coefficient 
 $\sigma_\xi (pc_+^\alpha + (1-p)c_-^\alpha)^{1/\alpha}$.
 \end{lemma}

\noi 
 The next result provides some further insight to the observed {\it independence phenomenon} for block--maxima.
 Namely, it turns out that the block--maxima of a heavy--tailed time series are always asymptotically independent, provided
 that they converge to a max--stable process. 

 \begin{lemma}\label{l:ai} Let $X=\{X(k)\}_{k\in\bbZ}$ be a heavy--tailed time series with marginal distributions 
 as in \refeq{F-tail}.  Suppose that \refeq{X_m} holds where the $Z(k)$'s are not assumed independent.

 If the limit time series $Z=\{Z(k)\}_{k\in\bbN}$ is multivariate max--stable, then it consists of i.i.d.\  random variables.
 \end{lemma}

\noi The proof is given in the Appendix.

\subsection{On the construction of confidence intervals}
  \label{s:conf_int}

In many applications, an uncertainty assessment about the estimated tail exponent is important, which requires the 
construction of confidence intervals.

 The literature is rather sparse for confidence intervals for the heavy tail exponent even in the case of independent data. 
 We are not aware of any general results on the asymptotic distribution of the Hill or the moment 
 estimator of $\alpha$ for dependent data. Theorem \ref{t:main} above suggests the following
 {\it asymptotic confidence interval} for $\alpha$ of level $\gamma,\ 0 < \gamma <1$:
 \beq\label{e:ci}
 {\Big(}
  (\widehat{H} - \what{H} z_{(1-\gamma)/2} \sqrt{c_w}/\sqrt{n_{j}})^{-1},
   (\widehat{H} + \what{H} z_{(1-\gamma)/2} \sqrt{c_w}/\sqrt{n_{j}})^{-1} {\Big)},
 \eeq
 where $z_{(1-\gamma)/2}$ is $(1+\gamma)/2-$quantile of the standard normal distribution, and where
 $c_w$ as in Theorem \ref{t:main}.  Here, as recommended in \cite{stoev:michailidis:taqqu:2006}, we use 
 the reciprocal of a symmetric confidence interval for $H$ to obtain one 
 for $\alpha = 1/H$ (see also \refeq{main-1}).
 
\begin{table}
\begin{tabular}{||c||c|c|c|c|c|c|c|c|c||}
  \hline
$ j_1 $ & $ 3 $  & $ 4 $  & $ 5 $  & $ 6 $  & $ 7 $  & $ 8 $  & $ 9 $  & $ 10 $  & $ 11 $ \\
\hline
\hline
$90\% $ c.i.\ $\phi  = 0.1 $  & $ 0.891 $ & $ 0.894 $ & $ 0.912 $ & $ 0.919 $ & $ 0.897 $ & $ 0.903 $ & $ 0.889 $ & $ 0.895 $ & $ 0.875 $\\
$ \phi  = 0.3 $  &  $ 0.759 $ &  $ 0.888 $ &  $ 0.914 $ &  $ 0.915 $ &  $ 0.899 $ &  $ 0.901 $ &  $ 0.889 $ &  $ 0.895 $ &  $ 0.875 $\\
$ \phi  = 0.5 $  &  $ 0.229 $ &  $ 0.772 $ &  $ 0.889 $ &  $ 0.915 $ &  $ 0.892 $ &  $ 0.899 $ &  $ 0.888 $ &  $ 0.895 $ &  $ 0.875 $\\
$ \phi  = 0.7 $  &  $ 0.000 $ &  $ 0.299 $ &  $ 0.801 $ &  $ 0.895 $ &  $ 0.895 $ &  $ 0.899 $ &  $ 0.887 $ &  $ 0.895 $ &  $ 0.875 $\\
$ \phi  = 0.9 $  &  $ 0.000 $ &  $ 0.000 $ &  $ 0.070 $ &  $ 0.641 $ &  $ 0.843 $ &  $ 0.890 $ &  $ 0.877 $ &  $ 0.890 $ &  $ 0.875 $\\
\hline
\hline
$95\% $ c.i.\ $\phi = 0.1 $  & $ 0.943 $ & $ 0.952 $ & $ 0.954 $ & $ 0.953 $ & $ 0.949 $ & $ 0.950 $ & $ 0.931 $ & $ 0.931 $ & $ 0.904 $\\
$ \phi = 0.3 $  & $ 0.844 $ & $ 0.940 $ & $ 0.952 $ & $ 0.953 $ & $ 0.949 $ & $ 0.950 $ & $ 0.931 $ & $ 0.931 $ & $ 0.904 $\\
$ \phi = 0.5 $  & $ 0.321 $ & $ 0.854 $ & $ 0.950 $ & $ 0.954 $ & $ 0.948 $ & $ 0.950 $ & $ 0.931 $ & $ 0.931 $ & $ 0.904 $\\
$ \phi = 0.7 $  & $ 0.000 $ & $ 0.395 $ & $ 0.872 $ & $ 0.946 $ & $ 0.944 $ & $ 0.950 $ & $ 0.931 $ & $ 0.931 $ & $ 0.904 $\\
$ \phi = 0.9 $  & $ 0.000 $ & $ 0.000 $ & $ 0.123 $ & $ 0.738 $ & $ 0.911 $ & $ 0.941 $ & $ 0.927 $ & $ 0.930 $ & $ 0.904 $\\
\hline
\hline
$99\% $ c.i.\ $\phi = 0.1 $  & $ 0.990 $ & $ 0.990 $ & $ 0.989 $ & $ 0.991 $ & $ 0.987 $ & $ 0.993 $ & $ 0.975 $ & $ 0.972 $ & $ 0.947 $\\
$ \phi = 0.3 $  & $ 0.946 $ & $ 0.985 $ & $ 0.990 $ & $ 0.991 $ & $ 0.987 $ & $ 0.992 $ & $ 0.975 $ & $ 0.972 $ & $ 0.947 $\\
$ \phi = 0.5 $  & $ 0.552 $ & $ 0.953 $ & $ 0.984 $ & $ 0.990 $ & $ 0.987 $ & $ 0.991 $ & $ 0.975 $ & $ 0.972 $ & $ 0.947 $\\
$ \phi = 0.7 $  & $ 0.000 $ & $ 0.642 $ & $ 0.959 $ & $ 0.981 $ & $ 0.988 $ & $ 0.990 $ & $ 0.974 $ & $ 0.972 $ & $ 0.947 $\\
$ \phi = 0.9 $  & $ 0.000 $ & $ 0.000 $ & $ 0.276 $ & $ 0.897 $ & $ 0.968 $ & $ 0.984 $ & $ 0.973 $ & $ 0.972 $ & $ 0.947 $\\
\hline\hline
\end{tabular}

\noindent\caption{\label{tab:ci_MAX_AR_1}{\small Coverage probabilities of the asymptotic confidence intervals 
 \refeq{ci} for $\alpha$ for max--AR(1) time series as in \refeq{max-X-AR} of length $2^{15}$.  Max self--similarity
 estimators $\what H = \what H(j_1,j_2)$ were used with $1\le j_1\le j_2$ and $j_2 = 15$.  Results for three confidence levels:
 $90\%$, $95\%$ and $99\%$ are shown for different values of $j_1$.}}

\end{table}

\begin{table}[t!]
\begin{tabular}{||c||c|c|c|c|c|c|c|c|c||}
  \hline
$ j_1 $  & $ 5 $  & $ 6 $  & $ 7 $  & $ 8 $  & $ 9 $  & $ 10 $  & $ 11 $ & $12$ & $13$ \\
\hline
\hline
$90\% $ c.i.\ $ \phi  = 0.1 $ &  $ 0.884 $ &  $ 0.909 $ &  $ 0.903 $ &  $ 0.907 $ &  $ 0.901 $ &  $ 0.914 $ &  $ 0.887 $ &  $ 0.902 $ &  $ 0.917 $ \\
$ \phi  = 0.3 $  &  $ 0.903 $ &  $ 0.906 $ &  $ 0.915 $ &  $ 0.888 $ &  $ 0.898 $ &  $ 0.910 $ &  $ 0.906 $ &  $ 0.907 $ &  $ 0.916 $  \\
$ \phi  = 0.5 $   &  $ 0.911 $ &  $ 0.908 $ &  $ 0.905 $ &  $ 0.905 $ &  $ 0.898 $ &  $ 0.906 $ &  $ 0.890 $ &  $ 0.902 $ &  $ 0.898 $ \\
$ \phi  = 0.7 $   &  $ 0.837 $ &  $ 0.885 $ &  $ 0.879 $ &  $ 0.898 $ &  $ 0.906 $ &  $ 0.908 $ &  $ 0.907 $ &  $ 0.906 $ &  $ 0.899 $ \\
$ \phi  = 0.9 $   &  $ 0.103 $ &  $ 0.735 $ &  $ 0.863 $ &  $ 0.888 $ &  $ 0.894 $ &  $ 0.909 $ &  $ 0.920 $ &  $ 0.909 $ &  $ 0.915 $ \\
\hline\hline
$95\% $ c.i.\ $ \phi  = 0.1 $  &  $ 0.945 $ &  $ 0.953 $ &  $ 0.950 $ &  $ 0.947 $ &  $ 0.947 $ &  $ 0.951 $ &  $ 0.946 $ &  $ 0.953 $ &  $ 0.959 $ \\
$ \phi  = 0.3 $  &  $ 0.956 $ &  $ 0.944 $ &  $ 0.953 $ &  $ 0.941 $ &  $ 0.942 $ &  $ 0.955 $ &  $ 0.946 $ &  $ 0.963 $ &  $ 0.956 $ \\
$ \phi  = 0.5 $  &  $ 0.949 $ &  $ 0.955 $ &  $ 0.956 $ &  $ 0.947 $ &  $ 0.935 $ &  $ 0.945 $ &  $ 0.947 $ &  $ 0.952 $ &  $ 0.933 $ \\
$ \phi  = 0.7 $  &  $ 0.894 $ &  $ 0.949 $ &  $ 0.939 $ &  $ 0.949 $ &  $ 0.939 $ &  $ 0.956 $ &  $ 0.954 $ &  $ 0.947 $ &  $ 0.947 $ \\
$ \phi  = 0.9 $  &  $ 0.163 $ &  $ 0.820 $ &  $ 0.935 $ &  $ 0.943 $ &  $ 0.934 $ &  $ 0.958 $ &  $ 0.959 $ &  $ 0.959 $ &  $ 0.957 $ \\
\hline\hline
$99\% $ c.i.\ $ \phi  = 0.1 $  &  $ 0.992 $ &  $ 0.993 $ &  $ 0.992 $ &  $ 0.994 $ &  $ 0.984 $ &  $ 0.989 $ &  $ 0.989 $ &  $ 0.991 $ &  $ 0.997 $ \\
$ \phi  = 0.3 $  &  $ 0.995 $ &  $ 0.991 $ &  $ 0.987 $ &  $ 0.993 $ &  $ 0.985 $ &  $ 0.992 $ &  $ 0.992 $ &  $ 0.992 $ &  $ 0.998 $ \\
$ \phi  = 0.5 $  &  $ 0.990 $ &  $ 0.996 $ &  $ 0.991 $ &  $ 0.997 $ &  $ 0.993 $ &  $ 0.986 $ &  $ 0.984 $ &  $ 0.988 $ &  $ 0.980 $ \\
$ \phi  = 0.7 $  &  $ 0.953 $ &  $ 0.990 $ &  $ 0.989 $ &  $ 0.992 $ &  $ 0.990 $ &  $ 0.994 $ &  $ 0.988 $ &  $ 0.994 $ &  $ 0.993 $ \\
$ \phi  = 0.9 $  &  $ 0.337 $ &  $ 0.933 $ &  $ 0.984 $ &  $ 0.990 $ &  $ 0.980 $ &  $ 0.995 $ &  $ 0.984 $ &  $ 0.989 $ &  $ 0.993 $ \\
\hline\hline
\end{tabular}

\caption{\label{tab:ci_2_MAX_AR_1}{\small Coverage probabilities of empirical confidence intervals based on 
 Theorem \ref{t:alpha_n-to-alpha_F} for $\alpha$ for max--AR(1) time series as in \refeq{max-X-AR} of length $2^{15}$.  Max self--similarity
 estimators $\what H = \what H(j_1,j_2)$ were used with $1\le j_1\le j_2$ and $j_2 = 15$.  Results for three confidence levels:
 $90\%$, $95\%$ and $99\%$ are shown for different values of $j_1$.}}
\end{table}

Tables \ref{tab:ci_MAX_AR_1} and \ref{tab:ci_2_MAX_AR_1} illustrate coverage probabilities of confidence intervals for $\alpha$,
based on Theorems \ref{t:main} and \ref{t:alpha_n-to-alpha_F}, respectively.  They are based on $1\, 000$ independent replications of
max--AR(1) time series $X = \{X(k)\}_{k\in\bbZ}$:
\beq\label{e:max-X-AR}
 X(k) := \phi X(k-1) \vee Z(k) = \bigvee_{i=0}^\infty \phi^i Z(k-i),\ \ k=1,\ldots,n,
\eeq
of size $n= 2^{15} = 32\, 768$ for different values of $\phi$.  Here the $Z(k)$'s are i.i.d.\  and $\alpha-$Fr\'echet with $\alpha=1.5$.
The coverage probabilities for $90\%,~95\%$ and $99\%$ levels of confidence are reported in each row, as a function of $j_1$.

Observe that when the data are closer to independent ($\phi=0.1$), the coverage probabilities match the nominal
values even for small $j_1$'s.  As the degree of dependence grows, larger values for $j_1$ are required
to achieve accurate coverage probabilities.  Nevertheless, even in the most dependent setting ($\phi=0.9$) 
the value of  $j_1=8$ in Table \ref{tab:ci_MAX_AR_1} yields very good results.
 
Observe that coverage probabilities in Table \ref{tab:ci_MAX_AR_1} deteriorate for very large scales
$j_1$.  This is due to the inadequacy of the normal approximation in Theorem \ref{t:main} in the presence
of a limited number of block--maxima.  For large $j_1$'s the regime described in Theorem \ref{t:alpha_n-to-alpha_F}
is more applicable.  Table \ref{tab:ci_2_MAX_AR_1} shows that the coverage probabilities based on
\refeq{p:alpha_n-to-alpha_F} are very accurate even for the largest scales $j_1=13$.
We obtained these confidence intervals by using a Monte Carlo method.  Namely, we approximate the
distribution of the statistics $\what \alpha_F$ based on $1,000$ independent paths of i.i.d.\  $1-$Fr\'echet variables, 
multiplied by the estimated $\what \alpha_n$'s. Although these confidence intervals are significantly slower
to compute than \refeq{ci}, they exhibit excellent coverage probabilities even for the largest scales $j_1$. 

In conclusion, the brief numerical experiments suggest that the confidence intervals in \refeq{ci} work well
in practice, even for dependent data, for judicious choice of scales $j_1$ and $j_2$.  The confidence
intervals based on Theorem \ref{t:alpha_n-to-alpha_F} on the other hand, work well for all sufficiently large scales,
where the asymptotic normality may not apply. Both types of confidence intervals are useful in practice.

\section{On the automatic selection of the cut--off scale $j_1$}
  \label{s:cut-off}
 
  In the ideal case of $\alpha-$Fr\'echet i.i.d.\ data, the max--spectrum plot of $Y_j$ is 
 linear in $j$.  When the distribution of the data is not Fr\'echet, or when the data
are dependent, then the max--spectrum is asymptotically linear,
as the scales $j$ tend to infinity. It is therefore important to select appropriately the range of large
scales $j$ for estimation purposes. In view of \refeq{Yj-sim}, one can always choose $j_2 = [\log_2 n]$ 
to be the largest available scale and hence,
the problem is reduced to choosing the scale $j_1, 1\le j_1<j_2$.  The estimator of $\alpha$ is then obtained
by performing a WLS or GLS linear regression of $Y_j$ versus $j,\ j_1\le j\le j_2$ (see \refeq{H-hat}). 
 
The ``cut-off'' parameter $j_1$ can be selected either by visually inspecting the max--spectrum or
through a data driven procedure.  In \cite{stoev:michailidis:taqqu:2006} an automatic
procedure for selecting the cut--off parameter was proposed, in the case of independent data,
whose main steps are briefly summarized next. We also demonstrate that it performs satisfactorily 
for dependent data.
The algorithm sets $j_2 := [\log_2 n]$ and $j_1:= \max\{1, j_2-b\}$, with $b=3\mbox{ or }4$ in practice
for moderate sample sizes. Next, $j_1$ is iteratively decreased until statistically significant deviations
from linearity of $Y_j,\ j_1\le j \le j_2$ are detected.  Namely, as $j_1 > 1$, at each iteration over
the scale $j_1$ the following two quantities are calculated 
$\what H_{\rm new} = \what H(j_1-1,j_2)$ and $\what H_{\rm old} = \what H(j_1,j_2)$.
Whenever the value of zero is {\em not} contained in a confidence interval centered at  
$(\what H_{\rm new} - \what H_{\rm old})$, the algorithm stops and returns the selected $j_1$ and
$\hat\alpha = 1/\what H_{\rm old}$; otherwise, it sets $j_1:=j_1-1$ and proceeds accordingly.
The construction of the confidence interval about $(\what H_{\rm new} - \what H_{\rm old})$ utilizes
the covariance matrix $\Sigma_1$ in Theorem \ref{t:main} which is the same as in the i.i.d.\  case, see
\cite{stoev:michailidis:taqqu:2006}. The asymptotic normality result suggests that the methodology in
the case of i.i.d.\  data applies asymptotically to dependent data, for moderately large scales $j_1$.
Alternatively, the results of Theorem \ref{t:alpha_n-to-alpha_F} may be used to suitably correct the confidence 
intervals on the largest scales $j_1$.  We did not implement this method, since it is computationally 
demanding in practice.

\begin{figure}[ht!]
\begin{center}
\includegraphics[width=4.5in]{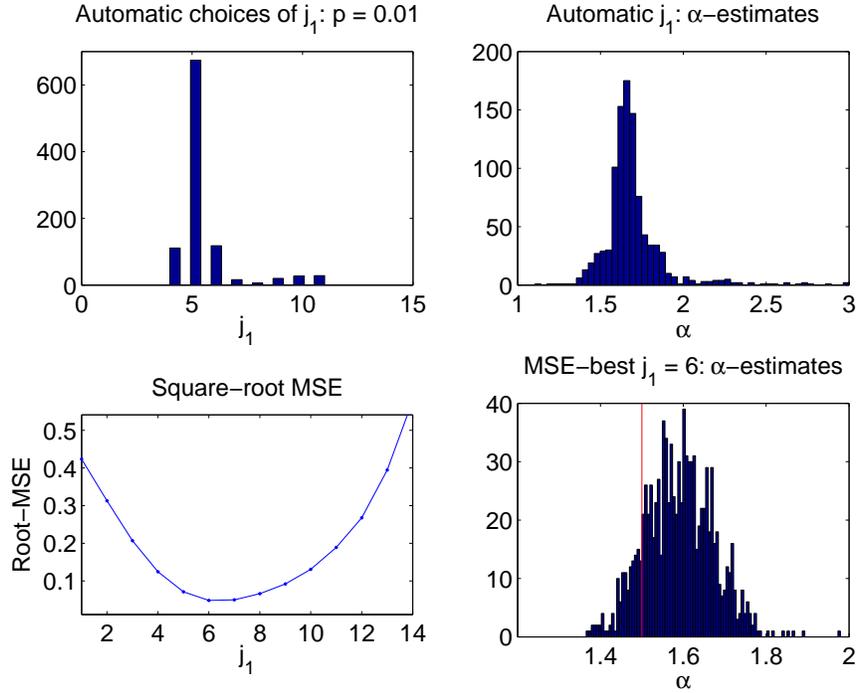}
{\caption{\label{fig:AR_frechet_cut_off} \small 
 The top--plot shows the histogram of automatically selected
 $j_1$ values for $1,000$ independent samples of size $N=2^{15}$ 
 from an exponential moving maxima $\alpha-$Fr\'echet process,
 $X = \{X(k)\}_{k\in\bbZ}$, defined as in \refeq{max-X-AR} with $\phi = 0.9$ and
 with i.i.d.\ $1.5-$Fr\'echet innovations.
 We used significance level  and back--start parameters
 are $p=0.01$ and $b=4$, respectively.  The top--right plot show
 the histogram of the resulting $\what\alpha=1/\what H$ estimates.
 The bottom--left plot shows estimates of the square root of the mean
 squared error (MSE) $\E (\what H- H)^2$ as a function
 of $j_1$.  The bottom--right plot contains a histogram of $\what \alpha$ 
 estimates obtained with the MSE--optimal choice of 
 $j_1=11$. }}
\end{center}
\end{figure}

 Figure \ref{fig:AR_frechet_cut_off} demonstrates the performance of the automatic
 selection procedure in the case of dependent data.  Even though the marginal
 distributions of $X$ are Fr\'echet, the dependence causes a knee in the
 max--spectrum plot (see, e.g.\ Figure \ref{fig:first-data-example-mspec}).  The automatic 
 selection procedure picks up this ``knee'' and yields reasonably unbiased and 
 precise {\it automatic estimates of $\alpha$} (see the top--right panel in Figure
 \ref{fig:AR_frechet_cut_off}).  Comparing the MSE plot and the histogram
 of the selected $j_1$ values, we see that over $70\%$ of the times the value 
 $j_1=5$ was chosen, which is close to the optimal value of $j_1=6$.  The histogram of the
 resulting automatic estimates of $\alpha$ (top--right panel) is similar 
 (with the exception of a few outliers) to the histogram of the estimators corresponding 
 to the MSE--optimal $j_1=6$ (bottom--right panel).

 Recall Table \ref{tab:ci_MAX_AR_1}, and observe that the case $\phi=0.9$ corresponds to
 the time series analyzed in Figure \ref{fig:AR_frechet_cut_off}.  The coverage 
 probabilities of the confidence intervals for $\alpha$ essentially match the nominal
 levels, for $j_1\ge 8$.  On the other hand the MSE--optimal value is $j_2=6$ 
 (Figure \ref{fig:AR_frechet_cut_off}) which is only slightly smaller than $j_1=8$.
 This can be contributed to the fact that the bias involved in the estimators at $j_1=6$,
 although comparable to their standard errors is significant and noticeably {\it shifts} the
 confidence interval.  As the scale $j_1$ grows, the bias quickly becomes negligible and the
 resulting confidence intervals become accurate.

\medskip
These brief experiments suggest that the automatic procedure is practical 
and works reasonably well in the case of dependent moving maxima time series.
Similar experiments for independent heavy--tailed data (not shown here)
indicate that the automatic selection procedure continues to perform well and chooses
values of $j_1$ close to the MSE--optimal ones, thus making it appropriate for use in 
empirical work. Nevertheless, a detailed study of its performance under a combination
of heavy--tailed distributions and dependence structures, as well as its sensitivity to
the choice of the back--start parameter $b$ and the level of significance $p$, is 
necessary and the subject of future work.

\section{Applications to Financial Data}
\label{s:data}

We analyze market transactions for two stocks -Intel (symbol INTC)  and Google (GOOG)-
using the max--spectrum.
The data sets were obtained from the {\it Trades and Quotes} (TAQ) data
base of {\it consolidated transactions} of the {\it New York Stock  Exchange} (NYSE) and
NASDAQ
(see \cite{wharton.data.service}) and  include the following
information about every single trade of the underlying stock:
{\it time of transaction} (up to seconds), {\it price} (of the share)
and {\it volume} (in number of shares).  In our analysis, we focus on  the traded volumes
of
the two stocks for November 2005, that could provide information  about the respective
sector's,
as well as the market's economic conditions (\cite{lo:wang:2000}).

A ubiquitous feature of the volume data sets is the presence of  heavy, Pareto type tails,
as can be seen in Figure \ref{fig:Google_Nov_7}. Specifically, the  top panel shows
transaction volumes
for the Google stock on November 7, 2005, while the bottom panels
show the Hill and the max--spectrum plots, respectively. The tail  exponent,
estimated from the max--spectrum over the range of scales $(11,15)$  is $\what\alpha =
1.0729$.
The Hill plot indicates heavy--tail exponent estimates between $1.5$  and $2$, which
correspond to the slope of the max--spectrum over the range of scales  $(1,10)$.
The small dip in the Hill plot for very large order statistics (small  values of $k$) can be
related to the behavior of the max--spectrum for scales $(11,15)$.
Such behavior is typical for almost  all liquid stocks, as well
as the presence of non--stationarity and dependence. In order to
minimize the intricate non--stationarity effects, we focus here on
traded volumes within a day.  The max--spectrum yields consistent
tail exponent estimates even in the presence of dependence.  This  fact and the
robustness of the max--spectrum suggest that it may be safely used in various practical
scenarios involving heavy--tailed  data.  In Figure \ref{fig:Google_all},
we show the max self--similarity estimates of the tail exponents, for each of the 21 trading
days in November, 2005.   The max--spectra
of these 21 time series (not shown here) of trading volumes are  essentially linear.  This
confirms the validity of a heavy--tailed model for the data, valid  over a wide range of
time scales -- from seconds up to hours and days.  Further, at the  beginning and end of
the trading day, several large volume transactions are observed, as  documented in
\cite{hong:wang:2000}.  Nevertheless, the  trading activity of Google, remains
essentially linear over the period under study, with a few bumps at the largest scales 
due to diurnal effects and other non--stationarities.

\begin{figure}[h!]
\begin{center}
\includegraphics[width=4.5in]{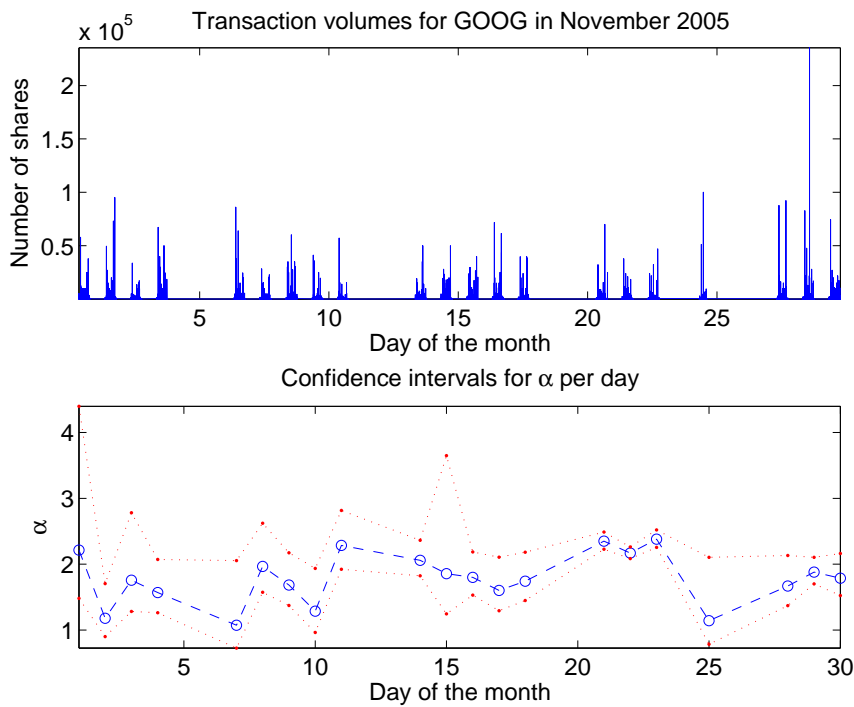}
{\caption{\label{fig:Google_all} \small
{\it Top panel:} traded volumes for the Google stock from the TAQ  data base of
consolidated
trades of NYSE and NASDAQ for the month of November, 2005.  The x-- axis and y--axis
correspond
to time and number of traded shares, respectively.  This is a high-- frequency data set,
where
each data point corresponds to the volume of a single transaction and  no temporal
aggregation
is performed.   The gaps of zeros in the data correspond to hours of  the day with no
trading and/or
weekends.
{\it Bottom panel:} estimated tail exponents (indicated by circles)  from the
max--spectrum and
their corresponding $95\%$ confidence intervals (indicated by broken  lines), based on the
asymptotic expression in \refeq{ci}.  Automatic selection
of the cut-off scale $j_1$ was done with $p=0.1$ and $b=3$ (see  Section \ref{s:cut-off}).
Every estimate was computed from a day worth of transaction volumes.
}}
\end{center}
\end{figure}

In Figure \ref{fig:Google_all}, the daily tail exponent estimates are  shown for the
Google
stock, which fluctuate between 1 and 2, along with pointwise  confidence intervals
(broken lines).
These estimates indicate that the tail exponent exhibits a  significant degree of
variability over
the period of a month, and that an infinite variance model may be  most appropriate for
modeling
trading volumes. For example, on November 7 (see Figure \ref {fig:Google_Nov_7}), the
estimate of $\alpha$ is nearly $1$, which may be due to the several  extremely large
peaks in
the volume data.  The upward knee in the max--spectrum of this data  set is likely
caused by these peaks.  The max--spectra on most other days are much  closer to linear
than the one in Figure \ref{fig:Google_Nov_7}.  Such correspondence  between the presence
of large peaks in the data and the behavior of the max--spectrum can be  used to identify
statistically significant fluctuations in the volume data.  Hence,  the max--spectrum plot
can be used not only to estimate $\alpha$, but also to detect changes  in the market.
We illustrate this last point next, by examining an unusual trading pattern in the Intel stock 
towards the end of  November, 2005.

\begin{figure}[h!]
\begin{center}
\includegraphics[width=4.5in]{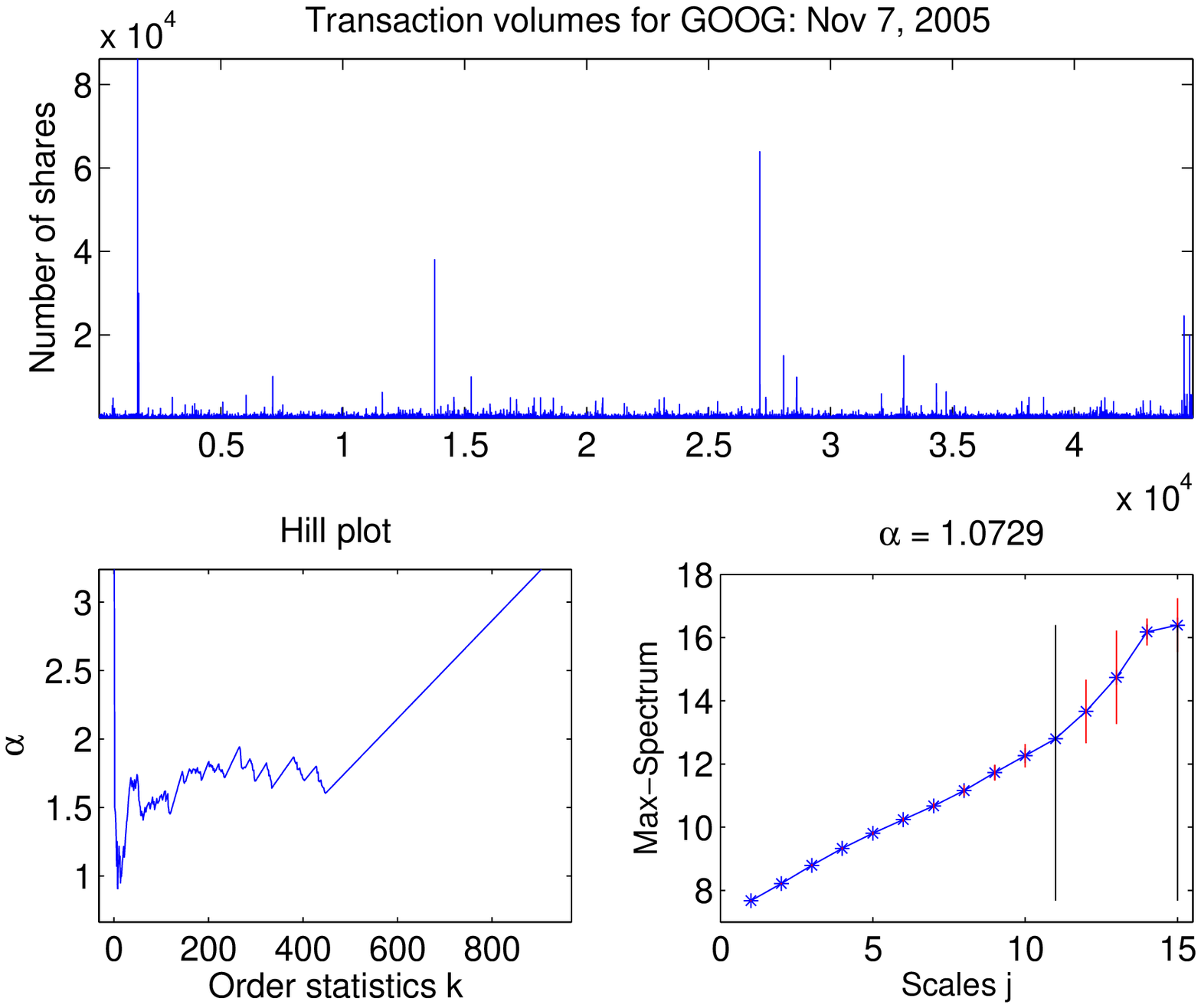}
{\caption{\label{fig:Google_Nov_7} \small
{\it Top panel:} the transaction volumes during the trading hours of  November 7, 2005. 
The
x--axis corresponds to the number of the transaction and the y--axis  to number of shares.
Note that about $50, 000$ transactions occurred on this day, which is  typical for the
Google stock.
Observe also the fairly classical heavy--tailed nature of the volume  data.
{\it Bottom panels:} the Hill plot (left) and the max--spectrum  (right) of the data. 
The Hill
plot is zoomed--in to a range where it is fairly constant and a tail  exponent between
$1.5$ and $2$
can be identified.  The max--spectrum reveals more: on large scales  the plot is steeper
than
on small scales with the tail exponent about $1$ on the range of  scales $(11,15)$ and
exponent
about $1.7$ on scales $(1,10)$.  The presence of a knee in the max-- spectrum plot
suggests different
behavior of the largest volumes on large time scales than on small  time scales and can
be contributed
to the several very large spikes of over $20, 000$ traded shares  (about 5 million US
dollars)
the top plot.
}}
\end{center}
\end{figure}

\begin{figure}[h!]
\begin{center}
\includegraphics[width=4.5in]{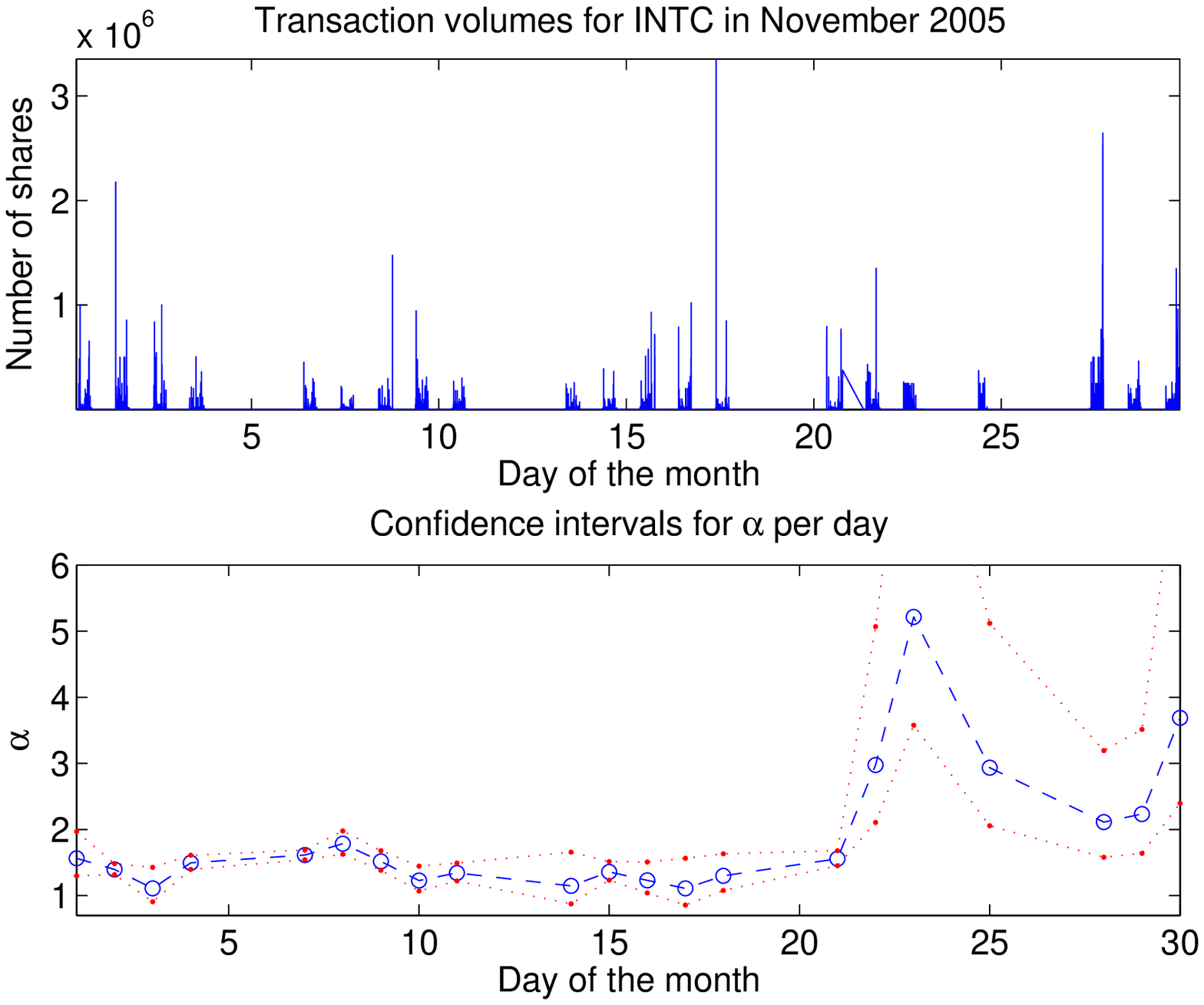}
{\caption{\label{fig:Intel_all}\small
This figure has the same format as Figure \ref{fig:Google_all}.  On  the top panel,
the traded volumes of the Intel stock for the month of November, 2005  are shown.
Observe that the tail exponent estimates on the bottom plot fluctuate  between 1.5 and 2
up to November 21.  On and after November 22, unusually high values  of $\alpha$ appear
(compare with the case of the Google stock in Figure \ref {fig:Google_all}).  This
is further analyzed in Figures \ref{fig:Intel_strange} and \ref {fig:Intel_normal}, below.
}}
\end{center}
\end{figure}

Figure \ref{fig:Intel_all} shows the max--spectrum estimates of the  tail exponents
for the traded volumes of the Intel stock for 21 trading days in  November 2005. Notice
that up to November 21,
the tail exponent is fairly constant, fluctuating between 1.2 and 2.   On November 22
(Tue)
and 23 (Wed), before the Thanksgiving holiday on November 24 (Thur),  the tail exponent
takes values
larger than 3 and 5, respectively.  This change is quite surprising  and it is
deemed significant by the corresponding confidence intervals.  A  closer look at the data
from November 23
(Figure \ref{fig:Intel_strange}) shows a changing but persistent pattern of trading as
compared to November 21; see for example Figure \ref{fig:Intel_normal}).

This behavior proves persistent and continues on November 25,  after the
Thanksgiving holiday. Moreover, no such behavior was  observed for
the Google data on any of the 21 trading days in November, 2005.
  Although trading of extremely large volumes occurs on November 23,
as seen in Figure \ref{fig:Intel_strange}, these trades are very regular and hence
inconsistent
with a heavy--tailed model.  Although regular in time, these large  transactions occur
on a time scale of several minutes, and hence the small scales of the  max--spectrum are
not affected by these peaks and behave as on a normal trading day  (see Figure
\ref{fig:Intel_normal}).
However, the large peaks dominate the larger scales $j$ and their  regularity makes
the max--spectrum essentially horizontal.  The Hill plot, shown on  the bottom--left
panel of
Figure \ref{fig:Intel_strange}, fails to pick up the unusual  behavior, since it suggests
values of $\alpha \approx 1$, which corresponds only to the
smallest portion of the max--spectrum, where $\what \alpha (7,11) =  1.0578 \approx 1$.

Our best guess is that this change in activity is related to the approval by the board
of directors of the Intel Corp.\ on November 10 of a program for a stock buy--back  worth of up to
25 billion US dollars;  (see, e.g.\ the Financial Times, London, on  Thursday November
11, page 27); hence, some of the delayed effects of the announcement of the program and 
market reaction to it are  demonstrated in the volume  activity discussed above.

\begin{figure}[h!]
\begin{center}
\includegraphics[width=4.5in]{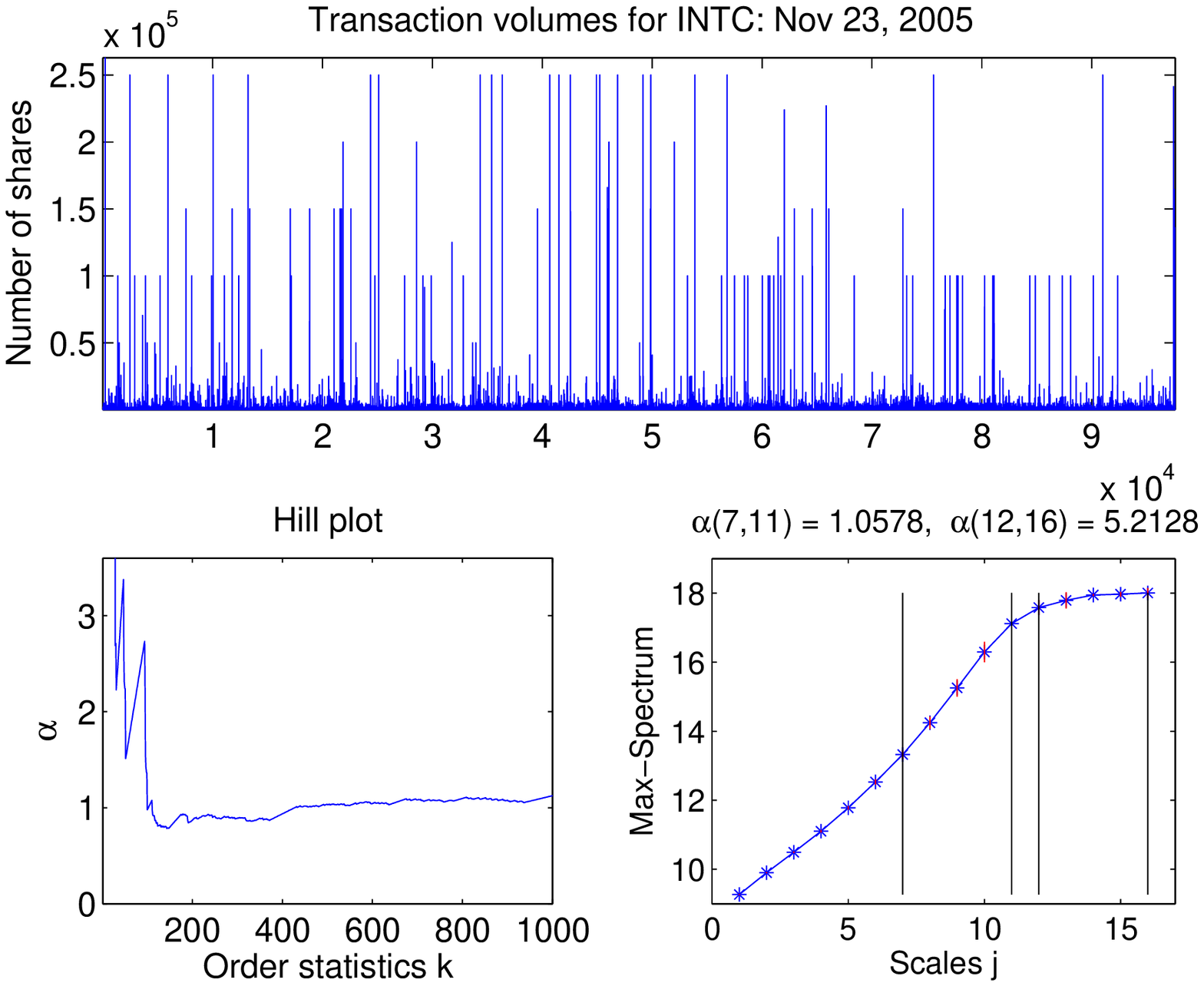}
{\caption{\label{fig:Intel_strange} \small
  {\it Top panel:} traded volumes of the Intel stock for
November 23, 2005.  Observe the regular occurrence of many very large  trades of
approximately the same sizes: 10, 000,\ 15, 000,\ 25, 000 and a few  of 20, 000 shares.
This is a very unusual behavior of the volume data, as compared to a  typical trading day
(see, e.g.\ Figure \ref{fig:Intel_normal}). {\it Bottom panels:} the  Hill plot and the
max--spectrum of the data.  Notice that the Hill plot fails to  identify the unusual
behavior of the data, whereas the max--spectrum flattens out, on  large scales due to
the regular non--heavy tailed behavior of the largest traded  volumes.  Once identified
on the
max--spectrum plot, one can perhaps read--off these details from the  volatile Hill plot
for very small values of $k$.  On small scales, where the regular large  transactions are not
frequent and
do not play a role, the max--spectrum yields tail exponents about $1 $.  This is
in line with the Hill plot.
}}
\end{center}
\end{figure}

\begin{figure}[h!]
\begin{center}
\includegraphics[width=4.5in]{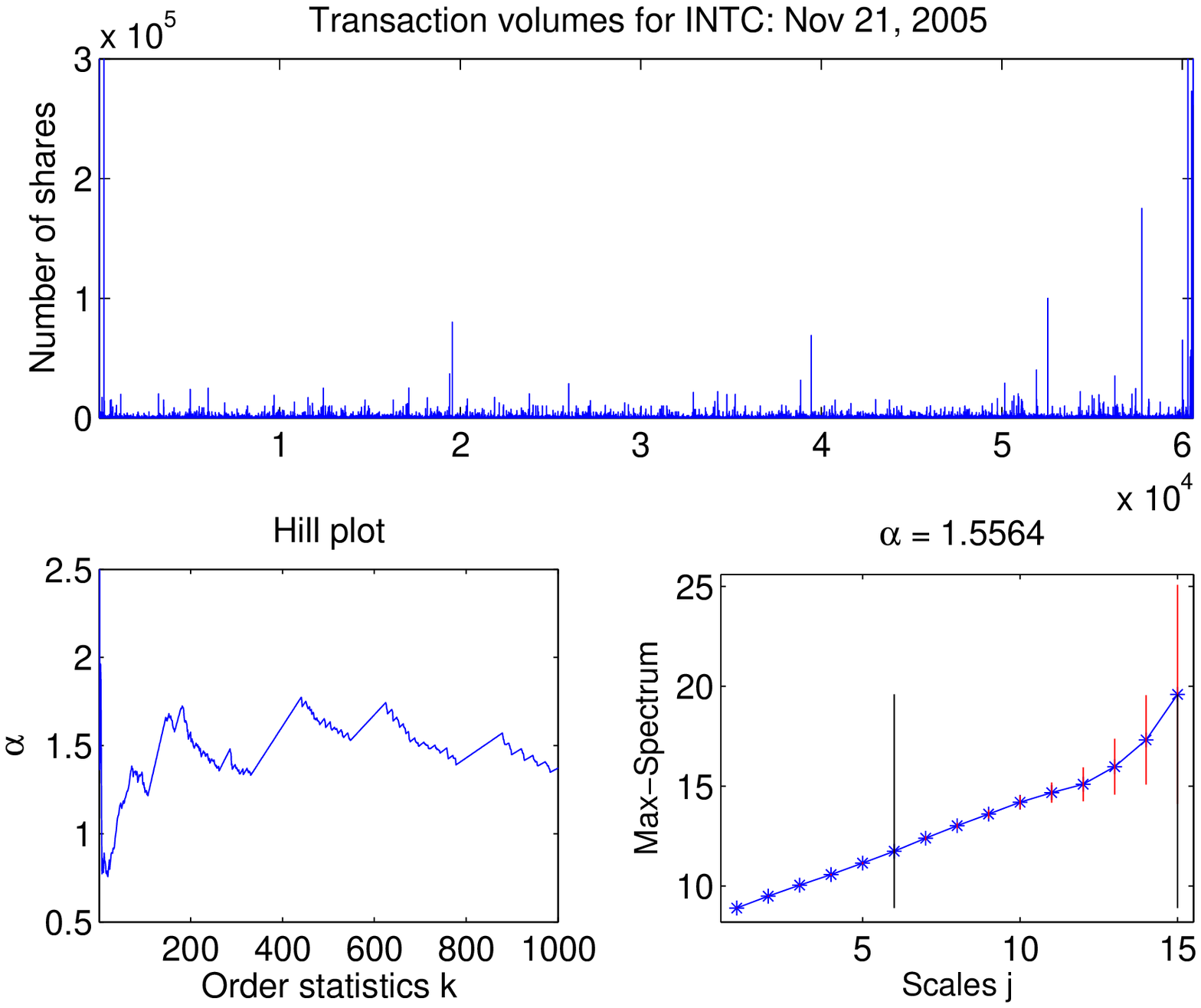}
{\caption{\label{fig:Intel_normal} \small
  This figure has the same format as Figure \ref{fig:Intel_strange}.
The top plot shows the volumes of INTC during November 21, 2005,  which as the
volumes of GOOG in Figure \ref{fig:Google_Nov_7}, behave like a  classical heavy--tailed
sample.
The Hill plot and the max--spectrum (bottom left and right panels,  respectively) identify
tail exponents around 1.5.  The cut off scale in the max--spectrum  plot was selected
automatically
with $p=0.1$ and $b=3$ (as in Figure \ref{fig:Intel_all}.   Notice  the volatile,
saw--tooth shape
of the Hill plot which is due to its non--robustness to deviations  from the Pareto model.
The max--spectrum is more robust and fairly linear with a small knee  on scale $j=12$,
which may be due to a few clusters of large volumes in the beginning  and at the end of
the trading
day.
}}
\end{center}
\end{figure}

\section*{Acknowledgments}  The authors would like to thank an anonymous referee for 
their detailed comments that led to an improved presentation of the material.  We are indebted 
to the referee for helping us strengthen the result of Lemma \ref{l:lin-proc}.  The work of
SS and GM was supported in part by NSF grant DMS-0806094.

\section{Appendix}
 \label{s:proofs}

\subsection{Rates of convergence for moment functionals of dependent maxima}
 \label{s:moments}

\begin{proposition} \label{p:moments}
Suppose that $f:(0,\infty) \to \mathbb{R}$ is an absolutely continuous function on
any compact interval $[a,b]\subset(0,\infty)$, and such that $f(x) = f(x_0) + \int_{x_0}^x f'(u) du,\ x>0$ for some (any) $x_0 >0$.

Let for some $m\in\bbR$ and $\delta>0$,
\begin{equation}\label{e:f-cond-1}
x^m|f(x)| + \esssup_{0<y\le x} y^m |f'(y)| \longrightarrow 0,\ \ \ \mbox{ as }x\downarrow 0,
\end{equation}
\begin{equation}\label{e:f-cond-2}
 x^{-\alpha}|f(x)| + x^{1+\delta} \esssup _{y\ge x} y^{-\alpha} |f'(y)| \longrightarrow 0,\ \ \ \mbox{ as }x\to\infty.
\end{equation}
Suppose also that the time series $X = \{X_n\}_{n\in\bbZ}$ satisfies Conditions \ref{cond:C1} and \ref{cond:C2},
where $c_1(x)$ is such that:
\begin{equation}\label{e:c1-cond}
 \int_{1}^\infty c_1(x) x^{-\alpha} |f'(x)| dx < \infty.
\end{equation}

\noi Then, $E |f(M_n)| <\infty,$ for all sufficiently large $n\in \mathbb{N}$, and for some $C_f>0$, independent of $n,$
\begin{equation}\label{e:Ef-order}
| \E f(M_n/n^{1/\alpha}) - \E f(Z)| \le C_f n^{-\beta},
\end{equation}
where $Z$ is an $\alpha-$Fr\'echet variable with scale coefficient $\sigma:= c_X^{1/\alpha}$.
\end{proposition}
\begin{proof} The proof is similar to the proof of Theorem 3.1 in \cite{stoev:michailidis:taqqu:2006}.
Indeed, as in the above reference, one can show that $\E |f(Z)|<\infty $ and $\E |f(M_n)|<\infty$, for all sufficiently large $n$.
Further, by using the conditions \refeq{f-cond-1} and \refeq{f-cond-2} and integration
by parts, we have that
\begin{equation}\label{e:prop:moments-1}
  \E f(M_n/n^{1/\alpha}) - \E f(Z) = \int_0^\infty (G(x) - F_n(x)) f'(x) dx,
\end{equation}
where $F_n(x):= \P\{M_n/n^{1/\alpha} \le x\}$ and $G(x) = \P\{Z\le x\}$. Since $F_n(x) = e^{-c(n,x)x^{-\alpha}},$ by the mean
value theorem, we have
\begin{eqnarray*}
|G(x) - F_n(x)| &=& | e^{- c_X x^{-\alpha}} - e^{-c(n,x) x^{-\alpha}}| \le |c(n,x) -  c_X| x^{-\alpha} e^{-\min\{\theta c_X,\, c(n,x)\} x^{-\alpha}} \\
                &\le& n^{-\beta} c_1(x)x^{-\alpha} {\Big(} e^{-c_2 x^{-(\alpha-\gamma)}} + e^{-\theta c_X x^{-\alpha}} {\Big)},
\end{eqnarray*}
where in the last inequality, we used Relations \refeq{C1} and \refeq{C2}.

Thus, by \refeq{prop:moments-1}, we have that
\begin{eqnarray}\label{e:prop:moments-2}
|\E f(M_n/n^{1/\alpha}) - \E f(Z)| &\le& n^{-\beta} \int_0^\infty c_1(x) x^{-\alpha} |f'(x)| {\Big(} e^{-c_2 x^{-(\alpha-\gamma)}}
             + e^{- c_X x^{-\alpha}} {\Big)} dx \nonumber\\
    &=:& n^{-\beta} {\Big(}\int_0^1 + \int_1^{\infty}{\Big)}.
\end{eqnarray}
The last integral is finite. Indeed, since the exponential terms above are bounded, Relation (\ref{e:c1-cond}) implies that
the integral ``$\int_1^\infty$'' is finite.  On the other hand, conditions \refeq{C1} and \refeq{f-cond-1} imply that,
$c_1(x) |f'(x)| = {\cal O}(x^{-R}),\ x\downarrow 0$, for some $R\in\bbR$.  However, for all $p>0$, we have
 $(e^{-c_2x^{-(\alpha-\gamma)}} + e^{- c_X x^{-\alpha}}) = o(x^{p}),\  x\downarrow 0$, since $\alpha-\gamma>0$.  This implies that
the integral in ``$\int_0^1$'' in (\ref{e:prop:moments-2}) is also finite.  This completes the proof of (\ref{e:Ef-order}).
$\Box$
\end{proof}

\medskip
\noindent{\sc Proof of Proposition \ref{p:log-moments}:} It is enough to show that the functions $f(x):= |\ln(x)|^p$ and
 $f(x):= (\ln(x))^k,\ p>0,\ k\in\bbN$ satisfy the conditions of Proposition \ref{p:moments}. In the first case, for example,
 $|f'(x)| = p x^{-1} |\ln(x)|^{p-1},\ x>0$.  Therefore, the assumption $\int_1^\infty c_1(x) x^{-\alpha-1+\delta} dx <\infty $
 implies (\ref{e:c1-cond}), since $ |\ln(x)|^{p-1} \le {\rm const}\, x^{\delta},$ for all $x\in [1,\infty)$.
 The conditions (\ref{e:f-cond-1}) and (\ref{e:f-cond-2}) are also fulfilled in this case, and hence Proposition \ref{p:moments} yields
 the desired order of convergence.  The functions $f(x) = (\ln(x))^k,\ k\in\bbN$ can be treated similarly.
$\Box$

\medskip
In the rest of this section we demonstrate that Conditions \ref{cond:C1} and \ref{cond:C2}
apply to a general class of moving maxima processes.

Let $\{Z_n\}_{n \in \mathbb{N}}$ be a sequence of i.i.d.\ random variables with the cumulative distribution
function $P\{ Z \leq z \} = F_Z(z)$.  As in \cite{stoev:michailidis:taqqu:2006}, we suppose that
\begin{equation}\label{e:C1_Z}
  F_Z(z) = \exp\{-c(z)z^{-\alpha} \},\ \ z>0,
\end{equation}
and impose two further conditions, analogous to Conditions \ref{cond:C1} and \ref{cond:C2}.

\begin{condition}\label{cond:C1-prime} There exists $\beta'>0$, such that
\begin{equation} \label{equ:cz_condition}
 |c(z) - c_Z| \le K z^{-\beta'},\ \ \ \mbox{ for all }  z >0,
\end{equation}
where $c_Z>0$ and $K\ge 0$.
\end{condition}

\begin{condition}\label{cond:C2-prime} $F_Z(0) = 0$ and for all $x>0$,
\begin{equation}\label{equ:cz_condition2}
 c(z) \ge c \min\{1, z^\gamma\}, \ \ \mbox{ for some }\gamma \in (0,\alpha),
\end{equation}
with $c>0$.
\end{condition}

\medskip
\noi Observe that (\ref{equ:cz_condition}) implies $c(z) \to c_Z,\ z\to\infty$, and in fact
$P\{Z> z\} = 1- F_Z(z) \sim c_Z z^{-\alpha},$ as $z \to \infty.$
Define now the moving maxima process $X = \{X_k\}_{k\in\bbZ}$:
\begin{equation} \label{equ:def_of_Xk}
  X_k := \max_{1\le i\le m} a_i Z_{k-i+1}, \ \ \ k\in{\mathbb Z},
\end{equation}
with some coefficients $a_i>0,\ i = 1,\ldots, m,$ and $m\ge 1$.
The following result shows that the process $X$ satisfies conditions Conditions \ref{cond:C1} \& \ref{cond:C2}.

\begin{proposition}\label{p:m_max}  If the $Z_n$'s satisfy Conditions  \ref{cond:C1-prime} and \ref{cond:C2-prime}, then 
the process $X=\{X_k\}_{k\in{\mathbb Z}}$ in (\ref{equ:def_of_Xk}) satisfies \refeq{F-tail}, Conditions \ref{cond:C1} and \ref{cond:C2}
with $\gamma$ as in (\ref{equ:cz_condition2}),
\begin{equation}\label{e:p:m_max-1}
 \sigma_0^\alpha = c_Z \sum_{i=1}^m a_i^\alpha,\ \ 
 \beta = \min\{1, \beta'/\alpha\} \ \ \  \mbox{ and }\ \ \ \ c_1(x) := {\rm const}(1 + x^{-\beta'}),
\end{equation}
where $\beta'$ is as in (\ref{equ:cz_condition}) and where $c_X:= c_Z \max_{1\le i\le m} a_i^\alpha$.
In particular, the extremal index of $X$ is $\theta = c_X/\sigma_0^\alpha =
{\max_{1\le i \le m} a_{i}^\alpha / \sum_{i=1}^m a_i^\alpha}.$
\end{proposition}
\begin{proof} We first derive the marginal distribution of the $X_k$'s. By (\ref{e:C1_Z}) and (\ref{equ:def_of_Xk}), we have
$$
 \P\{X_k \leq x \} = P\{ Z_k \leq x/a_1, \dots, Z_{k-m+1} \leq x/a_m\} = \exp\{ -\sum_{i=1}^{m}c(x/a_i)a_i^{\alpha}x^{-\alpha}\}.
$$
Thus, in view of (\ref{equ:cz_condition}), $c(x/a_i)\to c_Z,\ x\to\infty$, and hence, as $x\to\infty$
\begin{equation}\label{e:c_X}
 \P\{ X_k > x\} \sim \sigma_0^\alpha x^{-\alpha},\ \ \ \mbox{ where } \sigma_0^\alpha := c_Z \sum_{i=1}^{m} a_i^{\alpha}.
\end{equation}
We now focus on the maxima $M_n:= \max_{1\le i \le n} X_i.$ For $n > m$, and $x>0$, we have that
$F_n(x):= \P\{M_n/n^{1/\alpha}\le x\}$ equals
\begin{eqnarray}
  F_n(x) & = &  \P\{X_1 \leq n^{1/\alpha}x, \dots, X_n \leq n^{1/\alpha}x \} \nonumber \\
   & = & \P{\Big\{}\bigvee_{j=2-m}^{0} g_{j,m} Z_j \leq n^{1/\alpha}x, \ \bigvee_{j=1}^{n-m+1}a_{(1)}Z_j \leq n^{1/\alpha}x, \
      \bigvee_{j=0}^{m-2} h_j Z_{n-j} \leq n^{1/\alpha}x {\Big\}} \nonumber
\end{eqnarray}
where
$$
  a_{(1)} := \bigvee_{k=1}^{m}a_k, \ \ \ \ \ g_{j,m} = \bigvee_{k=2-j}^{m}a_k, \ \ \ \ \ h_j= \bigvee_{k=1}^{1+j}a_k.
$$
Therefore, by using the independence of the $Z_j$'s and Relation (\ref{e:C1_Z}), we get
$F_n(x) = \exp\{ -c(n,x)x^{-\alpha}  \},$ $x>0,$ where
\begin{equation} \label{equ:c(n,x)term}
 c(n,x) = \frac{1}{n} {\Big(} \sum_{j=2-m}^{0}c(n^{1/\alpha}x/g_{j,m}) g_{j,m}^{\alpha} + (n-m+1) a_{(1)}^{\alpha} c(n^{1/\alpha}x/a_{(1)})
  +  \sum_{j=0}^{m-2}c(n^{1/\alpha}x/h_j) h_j^{\alpha} {\Big)}.
\end{equation}

We will now show that Relation (\ref{e:C1}) holds with $\beta$ and $c_1(\cdot)$ as in (\ref{e:p:m_max-1}).
Let $c_X:= c_Z a_{(1)} = c_Z \max_{1\le i\le m} a_i^\alpha$. By (\ref{equ:c(n,x)term}), we have
\begin{eqnarray}
  |c(n,x) - c_X | & = &  |c(n,x) - c_Z a_{(1)}^{\alpha} | \nonumber \\
  & \le & \frac{1}{n} \sum_{j=2-m}^{0}{\Big|}c(n^{1/\alpha}x/g_{j,m}) - c_Z {\Big|} g_{j,m}^{\alpha} + \frac{(n-m+1)}{n}
    {\Big|}c(n^{1/\alpha}x/a_{(1)}) - c_Z{\Big|}  a_{(1)}^{\alpha} \nonumber \\
& & \ \ \ \ +  \frac{1}{n}\sum_{j=0}^{m-2} {\Big|}c(n^{1/\alpha}x/h_j) - c_Z {\Big|} h_j^{\alpha} + \frac{C}{n} =: A_1 + A_2 + A_3 + \frac{C}{n},
 \label{e:c(n,x)-theta_ c_X}
\end{eqnarray}
where the constant $C$ does not depend on $x$. In the last relation, we add and subtract the finite number of
$2(m-1)$ terms of the type $g_{j,m}^\alpha c_Z$ and $h_j ^\alpha c_Z$ and apply the triangle inequality.

Now, by applying Relation (\ref{equ:cz_condition}) to each one of the absolute value terms in $A_1$,
we obtain
\begin{equation}\label{e:A1-c(n,x)}
A_1 \le \frac{Ka_{(1)}^\alpha }{n}\sum_{j=2-m}^{0} n^{-\beta'/\alpha} x^{-\beta'} g_{j,m}^{\beta'} \le \frac{m-1}{n^{1+\beta'/\alpha}}
K a_{(1)}^{\alpha+\beta'} x^{-\beta'} = \frac{C_1}{n^{1+\beta'/\alpha}} x^{-\beta'},
\end{equation}
where the constant $C_1$ does not depend on $n$ and $x$ and where in the last inequalities we used that $g_{j,m}\le a_{(1)}$.
One obtains a similar bound for the term $A_3$ in (\ref{e:c(n,x)-theta_ c_X}):
\begin{equation}\label{e:A3-c(n,x)}
 A_3 \le \frac{C_3}{n^{1+\beta'/\alpha}} x^{-\beta'},
\end{equation}
where the constant $C_3$ does not depend on $n$ and $x$.

Now, for the term $A_2$ in (\ref{e:c(n,x)-theta_ c_X}), we also have by (\ref{equ:cz_condition}) that
\begin{equation}\label{e:A2-c(n,x)}
 A_2 \le \frac{n-m+1}{n} K a_{(1)}^{\alpha + \beta'} x^{-\beta'} n^{-\beta'/\alpha} \le \frac{C_2}{n^{\beta'/\alpha}} x^{-\beta'},
\end{equation}
where the constant $C_2$ does not depend on $n$ and $x$.

By combining the bounds in (\ref{e:A1-c(n,x)}) -- (\ref{e:A2-c(n,x)}), for the terms in (\ref{e:c(n,x)-theta_ c_X}), we obtain
$$
 |c(n,x) - c_X | \le \frac{(C_1+C_3)}{n^{1+\beta'/\alpha}} x^{-\beta'} + \frac{C_2}{n^{\beta'/\alpha}} x^{-\beta'} + \frac{C}{n},
$$
which shows that (\ref{e:C1}) holds with $c_1(x) = {\rm const}\, (1+x^{-\beta'}),$ where $\beta:= \beta'/\alpha$.

We now show that (\ref{e:C2}) holds. Since (\ref{e:C2}) involves a lower bound, we can ignore the two positive sums in
(\ref{equ:c(n,x)term}).  Recall (\ref{equ:cz_condition2})  and note that
$
  c(n^{1/\alpha}x/a_{(1)}) \ge
   c_2' \min\{1,  (n^{1/\alpha}x/a_{(1)})^{\gamma} \}.
$
Since, for sufficiently large $n$, $n^{1/\alpha}> a_{(1)}$, and $(n^{1/\alpha}x/a_{(1)})^{\gamma} \geq x^{\gamma},$
we obtain
$
 c(n^{1/\alpha}x/a_{(1)}) \geq c_2' \min\{1, x^{\gamma} \}.
$
Therefore, by (\ref{equ:c(n,x)term}), since for all sufficiently large $n$, $(n-m+1)/n\ge 1/2$, we have
$
  c(n,x) \geq c_2 \min\{1, x^{\gamma} \},
$
where $c_2 = a_{(1)}^{\alpha}c_2'/2$. This implies (\ref{e:C2}) and completes the proof of the proposition.
$\Box$
\end{proof}

\subsection{Auxiliary lemmas}

\noi The next three lemmas were used in the proof of Theorem \ref{t:main}.

\begin{lemma} \label{l:m-dep-1} Under the conditions of Theorem \ref{t:main}, for all $j > \log_2 m$,
we have
$$
 {\rm Var}(Y_j -  \widetilde Y_j ) \le \frac{3}{n_j} {\rm Var}(\log_2 (D(j,1)/ \widetilde{D}(j,1))).
$$
\end{lemma}
\begin{proof} For notational simplicity, let $\xi_k:=  \log_2(D(j,k)/ \widetilde{D}(j,k)),\ k=1,\ldots,n_j$.
We have, by the stationarity of $\xi_k$ in $k$, that
$$
{\rm Var}(Y_j - \widetilde{Y_j}) = \frac{1}{n_j} {\rm Var}(\xi_1)  + \frac{2}{n_j^2} \sum_{k=1}^{n_j-1} (n_j-k)
 {\rm Cov}(\xi_{k+1},\xi_1).
$$
Note that $\xi_{k+1} = \log_2(D(j,1+k)/ \widetilde{D}(j,1+k))$ and $\xi_1= \log_2(D(j,1)/ \widetilde{D}(j,1))$
are independent if $k>1$.  Indeed, this follows from the fact that the process $X$ is $m-$dependent, and since $\xi_{k+1}$
and $\xi_1$ depend on blocks of the data separated by at least $2^j > m$ lags.  Therefore, only the
lag--1 covariances in the above sum will be non--zero and hence
$$
{\rm Var}(Y_j - \widetilde{Y_j}) \le \frac{1}{n_j}  {\rm Var}(\xi_1)
+ \frac{2}{n_j} {\Big|}{\rm Cov}(\xi_{2},\xi_1) {\Big|} \le \frac{3}{n_j} {\rm Var}(\xi_1),
$$
since by the Cauchy--Schwartz inequality we have $|{\rm Cov}(\xi_2,\xi_1)| \le {\rm Var}(\xi_2)^{1/2} {\rm Var}(\xi_1)^{1/2} = {\rm Var}(\xi_1)$.
This completes the proof of the lemma.
$\Box$
\end{proof}

\begin{lemma} \label{lemma:convergence_of_ratio_of_BM_to_degenrate_RV} Under the conditions of Theorem \ref{t:main},
for any fixed $k$, we have ${D(j,k) / \widetilde{D}(j,k)} \stackrel {P}{\longrightarrow} 1,$ as $j\to\infty.$
\end{lemma}
\begin{proof} Let $\delta \in (0,1/\alpha)$ be arbitrary and observe that
\begin{equation}\label{e:D-D-tilde-ineq}
\P\{ D(j,k)/\widetilde D(j,k) < 1 \} = \P \{ R > \widetilde D(j,k)\}
\le \P \{ R > 2^{j\delta} \} + \P\{ 2^{j\delta} > \widetilde D(j,k)\},  \end{equation}
where $R = \max_{1\le i \le m} X_{2^{j}(k-i)+1}$.  Now, by stationarity,
$$\P\{ R > 2^{j\delta}\} = \P\{ \max_{1\le i \le m} X_i > 2^{j\delta}\} \to 0,
\ \ \mbox{ as }j\to\infty.
$$
On the other hand, Relation \refeq{C1} implies that $2^{-j/\alpha} \widetilde D(j,k) \stackrel{d}{\to}
Z,$ as $n\to\infty$, where $Z$ is a non--degenerate $\alpha-$Fr\'echet variable. Thus, since $\delta\in (0,1/\alpha)$,
we have that
$$
\P\{ 2^{j\delta} > \widetilde D(j,k)\} \to 0,\ \ \ \ \mbox{ as } j\to\infty.
$$
The last two convergences and the inequality (\ref{e:D-D-tilde-ineq})
imply that $\P\{ D(j,k)/\wtilde D(j,k) <1\} \to 0,\ j\to\infty$.  Since trivially $\P\{ D(j,k)/\wtilde D(j,k) >1\} =1$, we obtain
$D(j,k)/\wtilde D(j,k)$ converges in distribution to the constant $1$, as $j\to\infty$. 
This completes the proof since convergence in distribution
to a constant implies convergence in probability.
$\Box$
\end{proof}

\begin{lemma} \label{lemma:uniform_integrable_of_logratios}
The set of random variables
${\Big|} \log_2 {\Big(} \frac{D(j,k)}{ \widetilde{D}(j,k)} {\Big)} {\Big|}^p,\ \ \ j, k\in\bbN $ 
is uniformly integrable, for all $p > 0$, where $D(j,k)$ and $\wtilde D(j,k)$ are as in Theorem \ref{t:main}.
\end{lemma}
\begin{proof}  Let $q>p$ be arbitrary. By using the inequality $|x+y|^q \le 2^{q} (|x|^q + |y|^q),\ \ x,y\in\bbR$,
we get
$$
\E {\Big|}\log_2 \frac{D(j,k)}{ \widetilde{D}(j,k)} {\Big|}^q \le 2^{q} \E |\log_2 (D(j,k)/2^{j/\alpha})|^q  + 2^{q} \E |\log_2 (\wtilde D(j,k)/2^{j/\alpha})|^q.
$$
In view of Proposition \ref{p:log-moments}, applied to the block--maxima $D(j,k)$ and $\wtilde D(j,k)$, we obtain
$$
\E |\log_2 (D(j,k)/2^{j/\alpha})|^q = \E |\log_2(M_{2^j}/2^{j/\alpha})|^q \longrightarrow {\rm const},\ \ \mbox{ as }j\to\infty.
$$
Thus the set $\{\E |\log_2 (D(j,k)/2^{j/\alpha})|^q,\ j,k\in\bbN\}$ is bounded. We similarly have that the set
$\{\E |\log_2 (\wtilde D(j,k)/2^{j/\alpha})|^q,\ j,k\in\bbN\}$ is  bounded since $\log_2(2^{j} - m) \sim j,\ j\to \infty$, 
for any fixed $m$.

We have thus shown that
$$\sup_{j,k\in\bbN} \E {\Big|}\log_2 \frac{D(j,k)}{ \widetilde{D}(j,k)} {\Big|}^q <\infty,$$
for  $q>p$, which yields the desired uniform integrability. $\Box$
\end{proof}

\medskip
\noi We now present the proofs of Lemmas \ref{l:lin-proc} and \ref{l:ai} in Section \ref{s:d-cons}.

\medskip
\medskip\noi{\sc Proof of Lemma \ref{l:lin-proc}:}  
We start by noting that the results of Lemma 2.3, and Theorems 2.4 and 3.1 in   \cite{davis:resnick:1985} continue to hold 
for the linear process $X = \{X(k)\}_{k\in\bbZ}$, under the more general conditions in \refeq{lin-proc-cond}.  
The proofs of Theorems 2.4 and 3.1 in \cite{davis:resnick:1985} depend on the specific conditions in \eqref{e:lin-proc-cond} only 
through Lemma 2.3 and Relation (2.7) therein.  These two results (Lemma 2.3 and Relation (2.7)) are valid thanks to Lemma A.3
of \cite{mikosch:samorodnitsky:2000}.

Now, following the proof of Theorem 3.1 in  \cite{davis:resnick:1985}, introduce the 
map $T_r : M_p((0,\infty)\times \bbR\setminus\{0\}) \to \bbR^r,$
\beq\label{e:Tr}
 T_r {\Big(} \sum_{k=1}^\infty \epsilon_{(u_k,v_k)} {\Big)} : = 
 {\Big(}\vee_{u_k\in (0,1/r]} v_k, \cdots, \vee_{u_k\in(i/r, (i+1)/r]} v_k,\ \cdots, \vee_{u_k\in((r-1)/r, 1]} v_k{\Big)}.
\eeq
The map $T_r$ is simpler than than the map $T: M_p((0,\infty)\times \bbR\setminus\{0\}) \to D(0,\infty)$
considered in \cite{davis:resnick:1985}, where $D(0,\infty)$ denotes the Skorkhod space of {\it c\`adl\'ag} functions. 
The space $M_p$ of Radon point measures is equipped
with the topology of vague convergence, where a set $K\subset (0,\infty)\times \bbR\setminus\{0\}$ is compact if 
it is closed and bounded away from zero. We will argue below that the map $T_r$ is almost surely continuous when applied to
suitable Poisson random measures.

Proceeding as in the proof of  Theorem 3.1 in  \cite{davis:resnick:1985}, (by Theorem 2.4 (i) therein) we get
\beq\label{e:PRM}
 \sum_{k=1}^\infty \epsilon_{(k/m,m^{1/\alpha} X(k))} \Longrightarrow \sum_{i=0}^\infty \sum_{k=1}^\infty \epsilon_{(t_k,j_k c_i)}, 
\eeq
where '$\Rightarrow$' denotes weak convergence of point processes and where $\epsilon_{(t,j)}$ denotes a point measure
with unit mass concentrated at $(t,j)\in (0,\infty) \times \bbR\setminus\{0\}$. In \refeq{PRM},
$\{(t_k,j_k)\}_{k\ge 0}$ are the points of a Poisson random measure (PRM) with intensity measure 
$$\mu(dt,dx) = dx \times \lambda(dx),\ \ \mbox{ where }\ \
 \lambda(dx) = \alpha p x^{-\alpha-1}1_{(0,\infty)}(x) dx + \alpha (1-p) (-x)^{-\alpha-1}1_{(-\infty,0)}(x) dx,
$$
(recall the distribution of the $\xi_k$'s above and see (2.1) in \cite{davis:resnick:1985}).

Let now $m = \sum_{i=0}^\infty\sum_{k = 1}^\infty \epsilon_{(t_k,j_k c_i)}$ be the PRM in \refeq{PRM} and observe that 
$\P\{ m(\partial B) = 0 \} = 1$, where 
$B:= ((0,1/r]\times \bbR\setminus\{0\} ) \cup \cdots \cup ((r-1)/r,1]\times \bbR\setminus\{0\})$
is the set associated with the map $T_r$ in \refeq{Tr}, and where $\partial B$ denotes the boundary of $B$.
Indeed, this follows from the fact that the intensity measure $\mu(dt,dx)$ of the PRM $m$ does not charge 
with positive mass sets of Lebesgue measure zero.  The fact that $\P\{ m (\partial B) = 0\} =1$ shows that,
almost surely, the points $\{(t_k,j_k)\}$ do not lie on the boundary $\partial B$. Since the points of
discontinuity of $T_r$ are at only those measures in $M_p$ with atoms on $\partial B$, it follows that 
the map $T_r$ is almost surely continuous when applied to the realizations of the PRM $m$.
Therefore, the continuous mapping theorem (see e.g.\ Theorem 3.4.3 in \cite{whitt:2002}) yields:
$$
  T_r {\Big(} \sum_{k=1}^\infty \epsilon_{(k/m,m^{1/\alpha} X(k))} {\Big)}
  \stackrel{d}{\longrightarrow}   T_r {\Big(} \sum_{i=0}^\infty \sum_{k=1}^\infty \epsilon_{(t_k,j_k c_i)} {\Big)},
 \ \ \mbox{ as }m\to\infty,
$$
where
$$
 T_r {\Big(} \sum_{i=0}^\infty \sum_{k=1}^\infty \epsilon_{(t_k,j_k c_i)} {\Big)} 
  = {\Big(} \vee_{ t_k \in (0,1/r] }  \vee_{i=0}^\infty c_i j_k,\ \ldots,\ 
 \vee_{ t_k \in ((r-1)/r,1] } \vee_{i=0}^\infty c_i j_k {\Big)} =: (Z(1),\ldots, Z(r)).
$$
However, since the intervals $(0,1/r], (1/r,2/r],\ldots, ((r-1)/r, 1]$ in \refeq{Tr} do not overlap, 
the random variables $Z(1),\ldots,Z(r)$ are independent.  Moreover, the stationarity (in $t$) of the intensity of 
the PRM shows that the $Z(k)$'s are identically distributed.  Now, it remains to argue that the $Z(k)$'s have the
desired $\alpha-$Fr\'echet distribution.  This follows as in \cite{davis:resnick:1985}, since for $Z(1)$, for
example, we have:
$$
 Z(1)= \vee_{ t_k \in (0,1/r] } \vee_{i=0}^\infty c_i j_k = \vee_{ t_k \in (0,1/r] } (c_+ j_k \vee (-c_-) j_k ),
$$
which in fact equals the extremal process $Y(t)$ therein evaluated at $t=1/r$. 
$\Box$

\medskip
\noi{\sc Proof of Lemma \ref{l:ai}:}
 For multivariate max--stable distributions, pairwise independence implies independence (Ch.\ 5 in \cite{resnick:1987}).
 Thus, it suffices to show that $Z(1)$ and $Z(2)$ are independent.  The continuous mapping theorem implies that 
 $$
 \frac{1}{m^{1/\alpha}} (X_m(1) \vee X_m(2)) \stackrel{d}{\longrightarrow} Z(1)\vee Z(2),\ \ \mbox{ as }m\to\infty.
 $$
 We also have that $X_m(1)\vee X_m(2) = X_{2m}(1)$ and since $(2m)^{-1/\alpha} X_{2m}(1) \stackrel{d}{\to} Z(1),$ as $m\to\infty$,
 we obtain
 \beq\label{e:Z_12=2Z}
   Z(1) \vee Z(2) \stackrel{d}{=} 2^{1/\alpha}Z(1).
 \eeq

 In view of \refeq{F-tail}, the marginal distributions of $Z$ can only be $\alpha-$Fr\'echet.  Thus,
since $(Z(1),Z(2))$  is a max--stable vector, Proposition 5.11' in \cite{resnick:1987}, implies
\beq\label{e:Z_12}
 \P\{ Z_1\le x_1, Z_2 \le x_2\} = \exp{\Big\{} - \int_0^1 \frac{f_1^\alpha(u)}{x_1^{\alpha}} \vee \frac{f_2^\alpha(u)}{x_2^{\alpha}} du {\Big\}},
 \  \ \ x_1, x_2 >0.
\eeq
for some non--negative functions $f_1$ and $f_2$ such that $\int_0^1 f_i^\alpha (u) du <\infty,\ i=1,2$.
Thus, by \refeq{Z_12=2Z}, for all $x>0$:
\begin{eqnarray*}
 \P\{ Z(1) \vee Z(2) \le x\} & =& \exp\{ -x^{-\alpha} \int_0^1 f_1^\alpha(u)\vee f_2^\alpha(u) du\} \\
 & = &\P\{ Z_1 \le 2^{-1/\alpha}x\}
 = \exp\{ -2 x^{-\alpha} \int_0^1 f_1^\alpha (u) du\}.
\end{eqnarray*}
This, since by stationarity $\int_0^1 f_1^\alpha(u)du = \int_0^1 f_2^\alpha(u)du $, 
yields 
$$
 \int_0^1 f_1^\alpha(u) \vee f_2^\alpha (u) du = \int_0^1 f_1^\alpha(u)du + \int_0^1 f_2^\alpha(u)du.
$$
The last relation is valid if and only if the non--negative functions $f_1(u)$ and $f_2(u)$ have disjoint supports.
This fact, in view of \refeq{Z_12}, implies the independence of $Z(1)$ and $Z(2)$.
 $\Box$

\section* 
 \small
 

\end{document}